\documentclass[11pt,a4paper]{article}
\usepackage[lmargin=1.2in,rmargin=1.2in,bottom=1in,top=1in,twoside=False]{geometry}
\usepackage[table]{xcolor}
\usepackage{array, multirow} 
\usepackage{graphicx}
\usepackage{amsmath,amssymb,amsthm}
\usepackage{xspace, mathtools} 
\usepackage{amsthm}
\usepackage{hyperref}
\usepackage[capitalize,noabbrev]{cleveref}
\usepackage{amsfonts}
\usepackage[T1]{fontenc}
\usepackage[shortlabels]{enumitem}

\usepackage{inconsolata}
\usepackage{libertine}

\usepackage{tikz}
\usetikzlibrary{snakes}

\usepackage{bbding}
\usepackage[numbers]{natbib}

\newcommand*{\eg}{e.g.\@\xspace}
\newcommand*{\ie}{i.e.\@\xspace}
\newcommand{\set}[1]{\{#1\}}
\newcommand{\integers}[0]{\mathbb{Z}}
\newcommand{\Z}{\integers}
\newcommand{\naturals}[0]{\mathbb{N}}
\newcommand{\N}{\naturals}
\newcommand{\class}[1]{\textsf{#1}}
\newcommand{\poly}[1]{\textsf{poly}(#1)}
\renewcommand{\subset}{\subseteq}
\newcommand{\Oh}{\mathcal{O}}
\newcommand{\abs}[1]{\lvert #1 \rvert}
\newcommand{\bit}[1]{\mathrm{bitsize}(#1)}

\newcommand{\lex}{\prec_{lex}}

\newcommand{\len}[1]{\mathrm{len}(#1)}
\newcommand{\eff}[1]{\mathrm{eff}(#1)}
\newcommand{\bineff}[1]{\mathrm{eff}_b(#1)}
\newcommand{\unaeff}[1]{\mathrm{eff}_u(#1)}
\newcommand{\binguard}[1]{\mathrm{grd}_b(#1)}
\newcommand{\unaguard}[1]{\mathrm{grd}_u(#1)}

\renewcommand{\vec}[1]{{\bf #1}}
\newcommand{\oneguard}[1]{\mathrm{grd}_1(#1)}
\newcommand{\twoguard}[1]{\mathrm{grd}_2(#1)}
\newcommand{\configuration}[3]{#1(\vec{#2}#3)}
\newcommand{\config}[2]{\configuration{#1}{#2}{}} 
\newcommand{\covfig}[2]{\configuration{#1}{#2}{'}} 
\newcommand{\configvanilla}[2]{#1(#2)}
\newcommand{\run}[3]{#1\xrightarrow{#2}#3}
\newcommand{\PreserveBackslash}[1]{\let\temp=\\#1\let\\=\temp}
\newcolumntype{C}[1]{>{\PreserveBackslash\centering}p{#1}}
\newcolumntype{R}[1]{>{\PreserveBackslash\raggedleft}p{#1}}
\newcolumntype{L}[1]{>{\PreserveBackslash\raggedright}p{#1}}
\newcommand{\problemx}[3]{
\par\noindent\underline{\sc#1}\par\nobreak\vskip.2\baselineskip
\begingroup\clubpenalty10000\widowpenalty10000
\setbox0\hbox{\bf INPUT:\ }\setbox1\hbox{\bf QUESTION:\ }
\dimen0=\wd0\ifnum\wd1>\dimen0\dimen0=\wd1\fi
\vskip-\parskip\noindent
\hbox to\dimen0{\box0\hfil}\hangindent\dimen0\hangafter1\ignorespaces#2\par
\vskip-\parskip\noindent
\hbox to\dimen0{\box1\hfil}\hangindent\dimen0\hangafter1\ignorespaces#3\par
\endgroup}
\definecolor{lightpink}{RGB}{253, 164, 229}
\definecolor{lightblue}{RGB}{171, 236, 255}

\usepackage{cleveref}
\usepackage{subcaption}
\newtheorem{theorem}{Theorem}[section]
\crefformat{theorem}{#2Theorem~#1#3}
\Crefformat{theorem}{#2Theorem~#1#3}
\newtheorem{lemma}[theorem]{Lemma}
\crefformat{lemma}{#2Lemma~#1#3}
\Crefformat{lemma}{#2Lemma~#1#3}

\crefformat{conjecture}{#2Conjecture~#1#3}
\Crefformat{conjecture}{#2Conjecture~#1#3}

\newcommand{\newtheoremwithcrefformat}[2]{%
  \newtheorem{#1}{#2}[section]%
  \crefformat{#1}{##2\MakeUppercase#1~##1##3}%
  \Crefformat{#1}{##2\MakeUppercase#1~##1##3}%
}

\crefformat{page}{#2page~#1#3}%
\Crefformat{page}{#2Page~#1#3}%
\crefformat{equation}{#2(#1)#3}%
\Crefformat{equation}{#2(#1)#3}%
\crefformat{figure}{#2Figure~#1#3}%
\Crefformat{figure}{#2Figure~#1#3}%
\crefformat{section}{#2Section~#1#3}
\Crefformat{section}{#2Section~#1#3}
\crefformat{chapter}{#2Chapter~#1#3}
\Crefformat{chapter}{#2Chapter~#1#3}
\crefformat{chapter*}{#2Chapter~#1#3}
\Crefformat{chapter*}{#2Chapter~#1#3}
\crefformat{part}{#2Part~#1#3}
\Crefformat{part}{#2Part~#1#3}
\crefformat{enumi}{#2(#1)#3}
\Crefformat{enumi}{#2(#1)#3}

\newtheoremwithcrefformat{proposition}{Proposition}
\newtheoremwithcrefformat{observation}{Observation}
\newtheoremwithcrefformat{corollary}{Corollary}
\newtheoremwithcrefformat{claim}{Claim}
\newtheoremwithcrefformat{example}{Example}
\newtheoremwithcrefformat{remark}{Remark}
\newtheoremwithcrefformat{definition}{Definition}

\begin{document}

\title{
    Coverability in 2-VASS with One Unary Counter is in NP
}

\author{
    Filip Mazowiecki
        \footnote{University of Warsaw, Warsaw, Poland, 
        \textsf{f.mazowiecki@mimuw.edu.pl}. 
        Supported by the ERC grant INFSYS, agreement no. 950398.} \and
    Henry Sinclair-Banks 
        \footnote{Centre for Discrete Mathematics and its Applications \& Department of Computer Science, University of Warwick, Coventry, UK,
        \textsf{h.sinclair-banks@warwick.ac.uk}. 
        Supported by EPSRC Standard Research Studentship (DTP), grant EP/T5179X/1.} \and
    Karol W\k{e}grzycki
        \footnote{Saarland University and Max Planck Institute for Informatics, Saarbr\"ucken, Germany,
        \textsf{wegrzycki@cs.uni-saarland.de}. 
        Supported by the ERC grant TIPEA, agreement no. 850979.}
}

\date{}
\maketitle
\begin{abstract}
Coverability in Petri nets finds applications in verification of safety properties of reactive systems.
We study coverability in the equivalent model: Vector Addition Systems with States (VASS).

A $k$-VASS can be seen as $k$ counters and a finite automaton whose transitions are labelled with $k$ integers.
Counter values are updated by adding the respective transition labels.
A configuration in this system consists of a state and $k$ counter values.
Importantly, the counters are never allowed to take negative values.
The coverability problem asks whether one can traverse the $k$-VASS from the initial configuration to a configuration with at least the counter values of the target.

In a well-established  line of work on $k$-VASS, coverability in 2-VASS is already \class{PSPACE}-hard when the integer updates are encoded in binary.
This lower bound limits the practicality of applications, so it is natural to focus on restrictions.
In this paper we initiate the study of 2-VASS with one unary counter.
Here, one counter receives binary encoded updates and the other receives unary encoded updates.
Our main result is that coverability in 2-VASS with one unary counter is in \class{NP}.
This improves upon the inherited state-of-the-art \class{PSPACE} upper bound.
Our main technical contribution is that one only needs to consider runs in a certain compressed linear form.

\end{abstract}
\begin{picture}(0,0)
\put(-50,-340)
{\hbox{\includegraphics[width=40px]{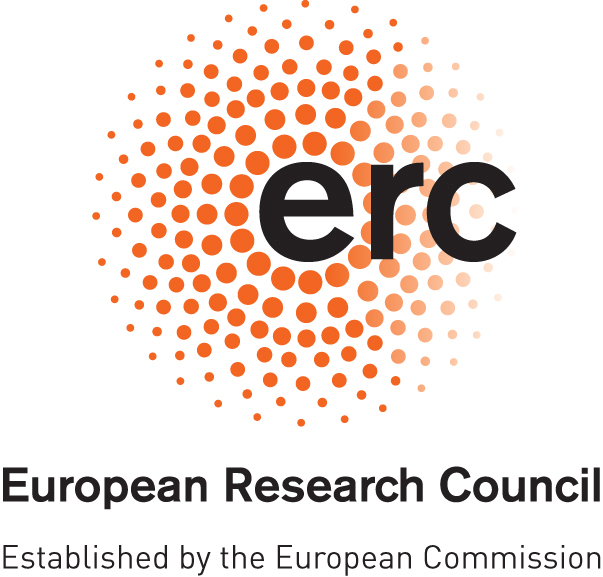}}}
\put(-60,-400)
{\hbox{\includegraphics[width=60px]{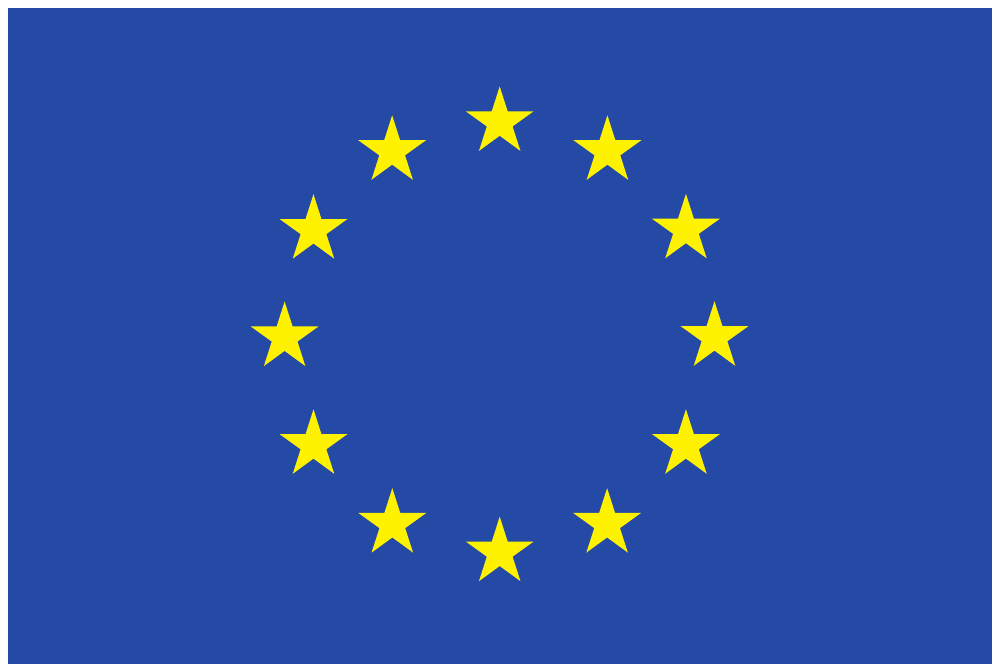}}}
\end{picture}

\section{Introduction} \label{sec:introduction}

Vector Addition Systems with States (VASS) are a well-studied class of infinite-state systems (see the survey~\cite{Schmitz16}).
These are finite automata with counters that can be updated, but are never allowed to take negative values.
Thus, a configuration consists of a state and a vector over the natural numbers.
The central decision problems are the reachability and coverability problems.
The reachability problem asks whether from a given start configuration one can reach the target configuration.
The coverability problem is the same except that the target configuration need not be reached exactly, counter values are allowed to be greater. 
Both problems are not only mathematically elegant, but they have interesting theoretical applications~\cite{BojanczykDMSS11} and implementations~\cite{BlondinHO21}.
Coverability is provably a simpler problem that is better suited for applications; reachability tools are mostly applied to coverability benchmarks~\cite{DixonL20}. 
Yet coverability has applications in the verification of safety conditions in reactive systems~\cite{EsparzaLMMN14,GantyM12}.
Such systems may require additional data structures to be accurately represented, like counters for example.
Safety conditions often boil down to whether a particular state can be reached as opposed to a particular configuration~\cite{ComonJ98}.

Coverability and reachability have been studied for decades. 
The equivalent model of Petri nets was introduced already in the sixties~\cite{Petri62}. 
For general VASS, Lipton proved in 1976 an \class{EXPSPACE} lower bound that applies to both coverability and reachability~\cite{Lipton76}. 
Two years later, Rackoff proved a matching \class{EXPSPACE} upper bound for coverability~\cite{Rackoff78}. 
Later in 1981, Mayr proved that reachability is decidable~\cite{Mayr84} without providing an upper bound for the algorithm. 
The construction was simplified by Kosaraju~\cite{Kosaraju82} and Lambert~\cite{Lambert92}, and a recent series of papers by Leroux and Schmitz ended in 2019 by proving an Ackermann upper bound~\cite{LerouxS19}. 
A matching Ackermann lower bound was published in 2021 by two independent groups~\cite{CzerwinskiO21,Leroux21}.

Plenty of attention has been given to VASS with fixed dimension, that is when the number of counters $k$ is invariable, denoted $k$-VASS.
For fixed dimension VASS it matters much whether the counter updates are encoded in unary or binary.
Already, Rackoff gives \class{NL} and \class{PSPACE} upper bounds for coverability in unary encoded and binary encoded $k$-VASS, respectively~\cite{Rackoff78}. 
The coverability problem where there are no counters is just directed graph reachability that is \class{NL}-complete~\cite{computational-complexity}.
Thus, coverability in unary encoded $k$-VASS is \class{NL}-complete, for every fixed $k$.
Coverability in binary encoded 1-VASS is in \class{NC}${}^2$~\cite{AlmagorCPSW20}, it can therefore be decided in deterministic polynomial time.
If there are two or more binary counters, coverability is \class{PSPACE}-hard~\cite{BlondinFGHM15} via a reduction from reachability in bounded one-counter automata that is \class{PSPACE}-complete~\cite{FearnleyJ13}.
Therefore, coverability in binary encoded $k$-VASS is \class{PSPACE}-complete for every $k \geq 2$.
See Figure~\ref{fig:coverability-bounds} for the complexities of coverability in VASS with a fixed number of unary and binary encoded counters.
This is all in striking contrast to the reachability problem in fixed dimension VASS, since reachability in 8-VASS is already known to be nonelementary~\cite{CzerwinskiO22}.

\begin{figure}[t]
    \begin{center}
    \begin{tabular}{c c c|C{37mm}|C{37mm}|C{37mm}|}
        & & \multicolumn{1}{c}{\hspace{1mm}} & \multicolumn{3}{c}{Number of unary counters}    \\
        & \multicolumn{2}{c}{\textcolor{white}{P = NP}}         & \multicolumn{1}{c}{$0$}   & \multicolumn{1}{c}{$1$}   & \multicolumn{1}{c}{$\geq 2$}  \\ \cline{4-6}
        \parbox[t]{2mm}{\multirow{5}{*}{\rotatebox[origin=c]{90}{\hspace{-6mm}Number of}}}   & 
        \parbox[t]{2mm}{\multirow{5}{*}{\rotatebox[origin=c]{90}{\hspace{-5mm}binary counters}}} 
        & \multirow{2}{*}{$0$}                                  & \multirow{2}{*}{\class{NL}-complete~\cite{computational-complexity}} & \multirow{2}{*}{\class{NL}-complete~\cite{ValiantP75}}& \multirow{2}{*}{\class{NL}-complete~\cite{Rackoff78}} \\ 
        & & & & & \\ \cline{4-6}    
        & & \multirow{2}{*}{$1$}                                &
        \multirow{2}{*}{in \class{NC}${}^2
        \subseteq\,$\class{P}~\cite{AlmagorCPSW20}} & \multirow{2}{*}{\textbf{in
        \class{NP}} [\emph{this paper}]} & \multirow{2}{*}{} \\
        & & & & & \\ \cline{4-6}    
        & & \multirow{2}{*}{$\geq2$}                            & \multirow{2}{*}{\class{PSPACE}-complete~\cite{BlondinFGHM15}} & \multirow{2}{*}{\class{PSPACE}-complete} & \multirow{2}{*}{\class{PSPACE}-complete~\cite{Rackoff78}} \\
        & & & & & \\ \cline{4-6}
    \end{tabular}
    \end{center}
    \vspace{-2mm}
    \caption{
    The complexities of coverability in VASS with a fixed number of unary and binary encoded counters.
    All \class{NL} lower bounds arise from the zero counters case, here coverability is directed graph reachability and that is well known to be \class{NL}-complete~\cite{computational-complexity}.
    In the case of one binary counter at least two unary counters, we are not aware of any non-trivial upper bounds below \class{PSPACE}.
    When there are at least two binary counters and any number of unary counters, coverability is \class{PSPACE}-complete. 
    The lower bound holds for 2-VASS with two binary counters~\cite{BlondinFGHM15} and the upper bound is given by Rackoff for any fixed dimension~\cite{Rackoff78}.
    Recall that coverability in general VASS, where the number of counters is not fixed, is \class{EXPSPACE}-complete~\cite{Rackoff78}.
    }
    \label{fig:coverability-bounds}
\end{figure}

There is a prominent line of work on 2-VASS with various encodings.
The seminal paper in 1979 of Hopcroft and Pansiot~\cite{HopcroftP79} shows reachability in 2-VASS is decidable, proving that the reachability set is effectively semi-linear. 
Moreover, in the same paper the authors show, by an example, that the 3-VASS reachability set need not be semi-linear. 
Later, this was improved as it was shown that for 2-VASS the reachability relation is effectively semi-linear~\cite{LerouxS04}. 
This proof shows that every 2-VASS can be characterised by a \emph{flat model}, \ie where the underlying finite automaton does not contain nested cycles. 
A more careful analysis of that paper, resulted in a \class{PSPACE} upper bound result for reachability in binary encoded 2-VASS~\cite{BlondinFGHM15}.
Since coverability in binary encoded 2-VASS is \class{PSPACE}-hard~\cite{BlondinFGHM15}, the authors were able to conclude that both coverability and reachability are \class{PSPACE}-complete.
Just as coverability demonstrated the difference encoding makes to complexity, so does reachability; later it was proved that reachability in unary encoded 2-VASS is \class{NL}-complete~\cite{EnglertLT16}.

\paragraph*{Our Results and Techniques.}
We consider the coverability problem for 2-VASS with one unary counter.
Here, updates of one counter are encoded in binary and the updates of the other are encoded in unary, see Figure~\ref{fig:example} for an example.
Notice that the unary counter need not be limited to polynomially bounded values.
Otherwise, the value of the unary counter could be encoded into the states for an instance of coverability in binary encoded 1-VASS.
Furthermore, we do not impose any restrictions on the initial and the target configurations, \ie both coordinates of these vectors are encoded in binary. 
Our main result is that coverability in 2-VASS with one unary counter is in \class{NP}.

Coverability in binary encoded $k$-VASS is \class{PSPACE}-complete, for $k \geq 2$.
The lower bound limits the practicality of applications.
Therefore, it is sensible to consider restricted variations and quantify their complexity.
We remark that coverability in fixed dimension VASS had widely-open complexity if there was exactly one binary counter and at least one unary counter.
See Figure~\ref{fig:coverability-bounds} for a summary of the known results. 

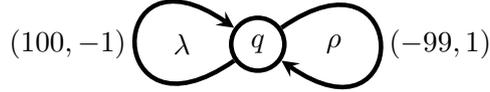
\begin{figure}[t]
    \centering
    \vspace{-6mm}
    \begin{tikzpicture}
        \node[circle, draw, minimum size = 7mm, line width = 0.6mm] (q) at (0, 0) {$q$};

        \path (q) edge[-stealth, loop above, looseness = 12, distance = 60, out = 215, in = 145, line width = 0.6mm] (q);
        \node (l) at (-1, 0) {$\lambda$};
        \node (l) at (-2.5, 0) {$(100, -1)$};

        \path (q) edge[-stealth, loop above, looseness = 12, distance = 60, out = 35, in = 325, line width = 0.6mm] (q);
        \node (l) at (1, 0) {$\rho$};
        \node (l) at (2.4, 0) {$(-99, 1)$};
    \end{tikzpicture}
    \vspace{-7mm}
    \caption{Example 2-VASS with one unary counter $V$.
    Consider the instance of coverability consisting of $V$, the initial configuration $q(0, 1)$, and the target configuration $q(0, 10)$.
    Consider the path $\pi = \lambda\rho \; \lambda\rho \cdots \lambda\rho \; \rho \cdots \rho$ which induces a run in $V$ from the initial configuration $q(0,1)$.
    There are $990$ repetitions of the pair of cycles $\lambda\rho$ to witness the configuration $q(990,1)$.
    The cycles alternate so both counters remain non-negative throughout the run.
    This is followed by $10$ iterations of the cycle $\rho$ so the configuration $q(0,11)$ is witnessed, achieving coverability of the target configuration $q(0,10)$.}
    \label{fig:example}
\end{figure}

The natural starting point is the characterisation of runs via \emph{linear path schemes}~\cite{BlondinEFGHLMT21}. 
Intuitively, the authors prove that if coverability or reachability holds then there is a witnessing path of a specific shape. 
Namely, all paths can be characterised by a bounded language defined by a regular expression of the form $\tau_0 \gamma_1^* \tau_1 \ldots \tau_{k-1} \gamma_k^* \tau_k$.
Here $\tau_0, \dots, \tau_k$ are paths that connect disjoint cycles $\gamma_1, \dots, \gamma_k$.
Since the language is bounded, checking if there is a path for a given expression essentially amounts to an instance of integer linear programming.
In particular, the authors argue that both $k$ and $\abs{\tau_0} + \abs{\gamma_1} + \abs{\tau_1} + \ldots + \abs{\tau_{k-1}} + \abs{\gamma_k} + \abs{\tau_k}$ are pseudo-polynomially bounded~\cite{BlondinEFGHLMT21}.
However, a polynomial bound would immediately yield an \class{NP} upper bound as such a regular expression can be guessed. 
Given that coverability in 2-VASS with two binary counters is \class{PSPACE}-hard~\cite{BlondinFGHM15}, we cannot simply directly apply the known results when dealing with 2-VASS with one binary and one unary counter.
In Section~\ref{sec:coverability}, we provide a detailed discussion and a difficult yet motivating example in Figure~\ref{fig:example-main}.

To overcome this problem, we show that coverability can be witnessed by paths in \emph{compressed linear form}. 
We relax the condition of the bounded language, by allowing to nest linear forms, provided that the exponents are fixed. 
Intuitively, an expression of the form $(\tau \gamma^* \tau')^*$ is still forbidden, but we allow for $(\tau \gamma^e \tau')^*$, where $e$ is fixed but can be exponentially large (encoded using polynomially many bits). Such a form easily provides an \class{NP} upper bound.

We rely on two crucial observations to prove that we can focus on paths in compressed linear form.
First, notice that the $*$ operation in a linear path scheme corresponds to iterating some cycle in the VASS. 
Since $\gamma_1, \ldots, \gamma_k$ need to be short, one naturally focuses on short cycles. 
The issue is that there are exponentially many cycles of polynomial size. 
In Section~\ref{sec:replacing} we prove that for coverability there are only polynomially many `optimal' cycles. 
In Section~\ref{sec:reshuffling} we deal with the problem when some cycle $\gamma$ occurs many times in a linear path scheme witnessing coverability, resulting in a polynomial bound on $k$, the width of the linear path scheme. 
Then we prove that, either we can merge some $\gamma_i$ and $\gamma_j$ thus reducing the width, or that there is a cycle that has positive effect on one counter and non-negative effect on the other counter. 
Intuitively, in the latter case, we can reduce the problem to coverability in 1-VASS by pumping such a cycle that forces one counter to take an arbitrarily large value. 
Moreover, such a cycle is witnessed by a linear path scheme. 
Since we need to pump this cycle, we require compressed linear forms to describe the repetitions of the cycle. 

We highlight that both our crucial observations rely on that we work with coverability, not reachability. 
We further highlight that we address these crucial observations through our technical contributions that often depend on the fact there is one unary counter.

\paragraph*{Further Related Work.}
Asymmetric treatment of the counters has been already considered for VASS. 
Recall that Minsky machines can be seen as VASS with the additional ability of zero-testing. 
For this model coverability is undecidable~\cite{Minsky67}, even with two counters. 
This raised natural questions of what happens where only one of the counters is able to be reset or tested for zero. 
This, and more generally, reachability in VASS with hierarchical zero-tests are known to be decidable~\cite{Reinhardt08}.
There is a further investigation into VASS with one zero-test~\cite{FinkelS10}.
Recently, work has appeared containing detailed analysis about 2-VASS where counters have different powers~\cite{FinkelLS18,LerouxS20}. 
Finally, one of the most famous open problems in the community is whether reachability is decidable for 1-VASS with a pushdown stack. 
For these systems, coverability is known to be decidable~\cite{LerouxST15}.
The best known lower bound is that coverability, thus reachability also, is \class{PSPACE}-hard~\cite{EnglertHLLLS21}.
Our model, 2-VASS with one unary counter, can be seen as 1-VASS with a singleton alphabet pushdown stack.

The complexity of reachability in binary encoded 3-VASS remains an intriguing open problem. 
It is \class{PSPACE}-hard, like in dimension two, and the only known upper bound is primitive recursive, but not even elementary~\cite{LerouxS19}.
Recent works on reachability in fixed dimension VASS~\cite{CzerwinskiLLP19,Czerwinski0LLM20,CzerwinskiO22} provide new examples and a better understanding of the VASS model. 
Interestingly, many techniques applied to fixed dimension VASS are very closely related to recent progress on the nonelementary and Ackermann lower bounds for general VASS~\cite{CzerwinskiLLLM21,CzerwinskiO21,Leroux21}.
We finally and additionally motivate coverability in VASS with one binary counter and (at least) one unary counter as an avenue for finding new techniques to approach VASS problems with.

\section{Preliminaries} \label{sec:preliminaries}

Given an integer $z \in \Z$ we denote $\bit{z} = \log_2(\abs{z} + 1) + 1$.
For a vector $\vec{v} \coloneqq (v_1,v_2)$ we use $(\vec{v})_1 \coloneqq v_1$ and $(\vec{v})_2 \coloneqq v_2$ to be the projections to the first and second coordinates, respectively. 
We use $|\vec{v}|_{\max} \coloneqq \max\{|v_1|,|v_2|\} + 1$ to denote the size of vector $\vec{v}$.
We write $\vec{v} \leq \vec{w}$ if the inequalities hold on each coordinate. 
We write $\vec{v} < \vec{w}$ if at least one of the inequalities is strict.

A \emph{2-VASS with one unary counter} $V = (Q, T)$ consists of a finite set of control \emph{states} $Q$ and a set of \emph{transitions} $T \subset Q \times \integers \times \set{-1, 0, 1} \times Q$. 
We shall refer to the first counter as the \emph{binary counter} and the second counter as the \emph{unary counter}. 
The size of $V$ is $\abs{V} = \abs{Q} + \sum_{(p,b,u,q) \in T} \bit{b}$.
With $\abs{V}_{\max} \coloneqq \abs{Q} + \abs{T}\cdot \abs{T}_{\max}$ we denote the total `pseudo-polynomial size' of the automaton, where $\abs{T}_{\max}$ denotes the maximum absolute value that occurs in the transitions.
Note that in a standard 2-VASS both counters are in binary, \ie the domain of updates for the second counter is also $\Z$.

A \emph{path} $\pi$ in $V$ is a, possibly empty, sequence of transitions $\pi = (t_i)_{i=1}^m$ such that $t_i = (q_{i-1}, b_i, u_i, q_i) \in T$.
A path is \emph{simple} if $q_0, \ldots, q_m$ are distinct.
A path is a \emph{cycle} if $q_0 = q_m$ and $m > 0$ (thus empty cycles are forbidden). 
We call it a $q_0$-cycle to emphasise the first and last state of the cycle. 
A cycle is \emph{simple} if $q_1, \ldots, q_m$ are distinct.
A cycle is \emph{short} if $m \leq \abs{Q}$.
The \emph{length} of a path is the number of transitions in the path, denoted $\len{\pi} = m$.
We write $\pi[i..j]$ to denote the path that is the subsequence of transitions $(t_i, \ldots, t_j)$ in $\pi$.

A \emph{configuration} $(p, \vec{u}) \in Q \times \naturals^2$, denoted $\config{p}{u}$, is a state paired with the current binary and unary counter values. 
A \emph{run} is a sequence of configurations $(\configuration{q_i}{v}{_i})_{i=0}^m$ such that $(q_{i-1}, (\vec{v}_i)_1 - (\vec{v}_{i-1})_1, (\vec{v}_i)_2 - (\vec{v}_{i-1})_2, q_i) \in T$.
A run can equivalently be defined by the sequence of configurations induced by following a path $\pi$ starting from an initial configuration $\configuration{q_0}{v}{_0}$. 
We denote this run $\run{\configuration{q_0}{v}{_0}}{\pi}{\configuration{q_m}{v}{_m}}$. 
We also write $\run{\configuration{q_0}{v}{_0}}{*}{\configuration{q_m}{v}{_m}}$ to indicate the existence of a run between two configurations.

In this paper we study the \emph{coverability} problem for VASS.
\vspace{2mm}
\problemx{VASS Coverability}
{A VASS $V = (Q,T)$ and two configurations $\config{p}{u}$ and $\config{q}{v}$.}
{Does $\run{\config{p}{u}}{*}{\covfig{q}{v}}$ hold, for some $\vec{v}' \geq \vec{v}$?}
\vspace{2mm}
Do note that the initial configuration $\config{p}{u}$ and the target configuration $\config{q}{v}$ have both the binary and unary components encoded as binary integers.
The \emph{reachability problem} for VASS---which we will not study in this paper---requires $\vec{v}' = \vec{v}$.

Consider a path $\pi = (t_i)_{i=1}^m$, where $t_i = (q_{i-1}, b_i, u_i, q_i)$. 
The \emph{effect} of $\pi$ is the sum of the counter updates, \ie the vector $\eff{\pi} \coloneqq \sum_{i=1}^m (b_i,u_i)$. We often focus on the two projections: the \emph{binary effect} $\bineff{\pi} \coloneqq \sum_{i=1}^m b_i$, and the \emph{unary effect} $\unaeff{\pi} \coloneqq \sum_{i=1}^m u_i$.

We say that a cycle $\gamma$ is \emph{monotone} if $\eff{\gamma} \ge \vec{0}$ or $\eff{\gamma} \leq \vec{0}$. 
Otherwise, we say that $\gamma$ is \emph{non-monotone}. 
Note the two variants of a non-monotone cycle: a \emph{positive-negative} cycle $\bineff{\gamma} > 0$ and $\unaeff{\gamma} < 0$, and a \emph{negative-positive} cycle $\bineff{\gamma} < 0$ and $\unaeff{\gamma} > 0$.

Let $\gamma$ be a cycle. 
Given $e \in \N$ we write $\gamma^e$ for the path obtained by $e$ repetitions of $\gamma$. 
We refer to $e$ as the \emph{exponent}.
A linear path scheme is a regular expression of the form $\tau_0 \gamma_1^* \tau_1 \cdots \tau_{k-1} \gamma_k^* \tau_k$, where the paths $\tau_0, \tau_1, \ldots, \tau_k$ connect disjoint cycles $\gamma_1, \ldots, \gamma_k$.
Note that a collection of cycles is disjoint if no two cycles have a common state.
Given $\ell = (\tau_0, \gamma_1, \tau_1, \dots, \tau_{k-1}, \gamma_k, \tau_k)$, we say the a path $\pi$ is in linear form $\ell$ if $\pi = \pi_\ell = \tau_0 \gamma_1^{e_1} \tau_1 \cdots \tau_{k-1} \gamma_k^{e_k} \tau_k$ for some exponents $e_1, \ldots, e_k$. 
Note that in this definition every path has a linear form, \eg $\tau_0 = \pi$ is valid.
To leverage the definition, we will ask whether paths are in a linear form of certain size.
The size of a linear form $\ell$ is $\sum_{i=0}^{k} \len{\tau_i} + \sum_{i = 1}^k \len{\gamma_i}$.
The size of $\pi_\ell$ is $ \sum_{i=0}^{k} \len{\tau_i} + \sum_{i = 1}^k \len{\gamma_i} + \sum_{i = 1}^k \bit{e_i}$, \ie includes the exponents.
We refer to $k$ as the \emph{width} of the linear form.

\section{Coverability in 2-VASS with One Unary Counter} \label{sec:coverability}

In this section we briefly discuss why the state-of-the-art techniques are not enough to prove that coverability in 2-VASS with one unary counter is in \class{NP}.
Blondin et al.~\cite{BlondinEFGHLMT21} show that for a given 2-VASS $V$ there exists a set of linear path schemes $S$ such that if $\run{\config{p}{u}}{*}{\config{q}{v}}$ in $V$, then there exists a path $\pi$ in a linear path scheme $\rho \in S$ such that $\run{\config{p}{u}}{\pi}{\config{q}{v}}$.
For every linear path scheme $\rho \in S$ the width of $\rho$, and therefore the width of every path, is bounded above by $\poly{|Q|,|T|_{\max}}$~\cite[Theorem 3.1]{BlondinEFGHLMT21}.
Such a path $\pi$ is not necessarily a polynomial size witness, as the width depends on $\abs{T}_{\max}$ polynomially. 
We provide an example of a 2-VASS with one unary counter where the width of every linear form $\ell$ for a path is exponential in the input size.
This demonstrates that the combinatorial structure of linear path schemes is not self-sufficient to show that there always exists a polynomial size witness of coverability.

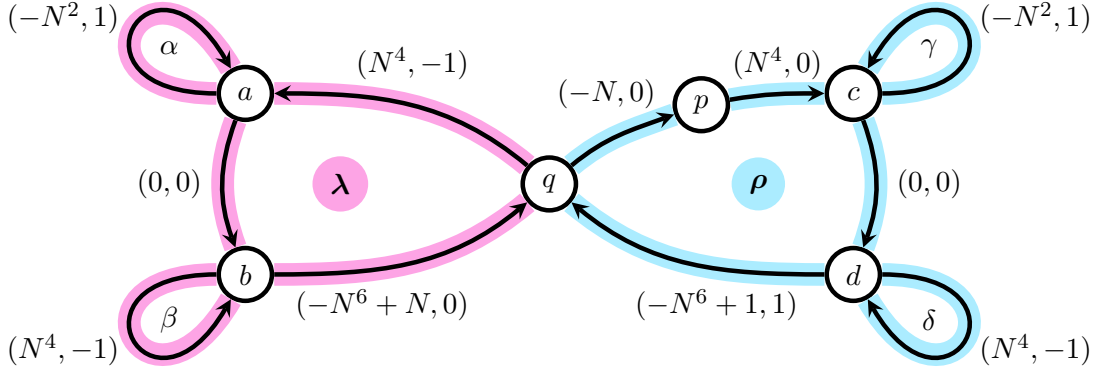
\begin{figure}[ht!]
    \begin{center}
        \begin{tikzpicture}
        \node[circle, draw, minimum size = 7mm, line width = 0.6mm] (q) at (0, 0) {$q$};
        \node[circle, draw, minimum size = 7mm, line width = 0.6mm] (p) at (2, 1.05) {$p$};
        \node[circle, draw, minimum size = 7mm, line width = 0.6mm] (a) at (-4, 1.2) {$a$};
        \node[circle, draw, minimum size = 7mm, line width = 0.6mm] (b) at (-4, -1.2) {$b$};
        \node[circle, draw, minimum size = 7mm, line width = 0.6mm] (c) at (4, 1.2) {$c$};
        \node[circle, draw, minimum size = 7mm, line width = 0.6mm] (d) at (4, -1.2) {$d$};

        \node[circle, minimum size = 2.5mm, fill = lightpink] (left) at (-2.75, 0) {$\boldsymbol\lambda$};
        \node[circle, minimum size = 2.5mm, fill = lightblue] (right) at (2.75, 0) {$\boldsymbol\rho$};

        \node (qa) at (-1.8, 1.6) {$(N^4, -1)$};
        \draw[-, >=stealth, line width = 3mm, color = lightpink] (q) to[out = 140, in = 0] (a);
        \draw[->, >=stealth, line width = 0.6mm] (q) to[out = 140, in = 0] (a);

        \node (ab) at (-5, 0) {$(0, 0)$};
        \draw[-, >=stealth, line width = 3mm, color = lightpink] (a) to[out = 250, in = 110] (b);
        \draw[->, >=stealth, line width = 0.6mm] (a) to[out = 250, in = 110] (b);

        \node (bq) at (-2.2, -1.6) {$(-N^6 + N, 0)$};
        \draw[-, >=stealth, line width = 3mm, color = lightpink] (b) to[out = 0, in = 220] (q);
        \draw[->, >=stealth, line width = 0.6mm] (b) to[out = 0, in = 220] (q);

        \node (qp) at (0.75, 1.2) {$(-N, 0)$};
        \draw[-, >=stealth, line width = 3mm, color = lightblue] (q) to[out = 40, in = 200] (p);
        \draw[->, >=stealth, line width = 0.6mm] (q) to[out = 40, in = 200] (p);

        \node (pc) at (3, 1.6) {$(N^4, 0)$};
        \draw[-, >=stealth, line width = 3mm, color = lightblue] (p) to[out = 10, in = 180] (c);
        \draw[->, >=stealth, line width = 0.6mm] (p) to[out = 10, in = 180] (c);

        \node (cd) at (5, 0) {$(0, 0)$};
        \draw[-, >=stealth, line width = 3mm, color = lightblue] (c) to[out = 290, in = 70] (d);
        \draw[->, >=stealth, line width = 0.6mm] (c) to[out = 290, in = 70] (d);

        \node (dq) at (2.2, -1.6) {$(-N^6 + 1, 1)$};
        \draw[-, >=stealth, line width = 3mm, color = lightblue] (d) to[out = 180, in = 320] (q);
        \draw[->, >=stealth, line width = 0.6mm] (d) to[out = 180, in = 320] (q);

        \node (ca) at (-6.4, 2.2) {$(-N^2, 1)$};
        \path (a) edge[loop above, looseness = 50, distance = 60, out = 180, in = 120, line width = 3mm, color = lightpink] (a);
        \path (a) edge[-stealth, loop above, looseness = 50, distance = 60, out = 180, in = 120, line width = 0.6mm] (a);
        \node (ga) at (-5, 1.8) {$\alpha$};

        \node (cb) at (-6.4, -2.2) {$(N^4, -1)$};
        \path (b) edge[loop above, looseness = 50, distance = 60, out = 180, in = 240, line width = 3mm, color = lightpink] (b);
        \path (b) edge[-stealth, loop above, looseness = 50, distance = 60, out = 180, in = 240, line width = 0.6mm] (b);
        \node (gb) at (-5, -1.8) {$\beta$};
  
        \node (cc) at (6.4, 2.2) {$(-N^2, 1)$};
        \path (c) edge[loop above, looseness = 50, distance = 60, out = 0, in = 60, line width = 3mm, color = lightblue] (c);
        \path (c) edge[-stealth, loop above, looseness = 50, distance = 60, out = 0, in = 60, line width = 0.6mm] (c);
        \node (gd) at (5, 1.8) {$\gamma$};

        \node (cd) at (6.4, -2.2) {$(N^4, -1)$};
        \path (d) edge[loop above, looseness = 50, distance = 60, out = 0, in = 300, line width = 3mm, color = lightblue] (d);
        \path (d) edge[-stealth, loop above, looseness = 50, distance = 60, out = 0, in = 300, line width = 0.6mm] (d);
        \node (gd) at (5, -1.8) {$\delta$};
    \end{tikzpicture}
    \vspace{-10mm}
    \end{center}
    \caption{
    Example 2-VASS with one unary counter $V$, where $N = 2^{n}$, where $n$ is an input parameter (thus making $N$ exponentially large).
    Consider the coverability instance with the initial configuration $q(0,1)$, and the target configuration $q(N, 1)$.
    Let $\lambda = t_{qa}\alpha^{N^2}t_{ab}\beta^{N^2}t_{bq}$ and $\rho = t_{qp}t_{pc}\gamma^{N^2}t_{cd}\delta^{N^2}t_{dq}$, where $t_{xy}$ is the transition from state $x$ to state $y$.
    Observe that $\eff{\lambda} = (N, -1)$ and $\eff{\rho} = (-N+1, 1)$, thus $\eff{\lambda\rho} = (1,0)$.
    It is easy to then see that $\run{q(0,1)}{(\lambda\rho)^N}{q(N,1)}$.
    Intuitively the cycles $\lambda$ and $\rho$ alternate so both counters remain non-negative throughout the run.
    In Appendix~\ref{app:coverability}, we prove that there does not exist a linear form of polynomial size for a path that induces a coverability run (Claim~\ref{clm:no-poly-linear-form-path-example}).
    }
    \label{fig:example-main}
\end{figure}

\paragraph*{Paths in Compressed Linear Form.}
Nevertheless, there is a natural way to succinctly describe the path presented in Figure~\ref{fig:example-main}. 
Let $\sigma = \lambda\rho$, and note that
\begin{equation*}
    \sigma^N = \left(t_{qa} \; \alpha^{N^2} \; t_{ab} \; \beta^{N^2} \; t_{bq}t_{qp}t_{pc} \; \gamma^{N^2} \; t_{cd} \; \delta^{N^2} \; t_{dq}\right)^{N}.
\end{equation*}
All paths and cycles are `small', and the bitsize of $N$ and $N^2$ are polynomial in $n$, so $\sigma$ itself is a path in linear form. 
We introduce the following generalisation of linear form paths that encapsulates the idea behind paths of this kind of arrangement.

\begin{definition}[Compressed linear form path]
A path $\pi$ is in \emph{compressed} linear form if $\pi = \rho_0 \sigma_1^{f_1} \rho_1 \cdots \rho_{k-1} \sigma_k^{f_k} \rho_k$ for some connected paths in linear form $\rho_0, \rho_1, \ldots, \rho_k$; cycles in linear form $\sigma_1, \ldots, \sigma_k$; and exponents $f_1, \ldots, f_k$.
The size of a compressed linear form path is the sum of the sizes of all $\rho_i$ and $\sigma_i$ (including the bitsize of their exponents) plus the bitsize of the exponents $f_i$.
\end{definition}

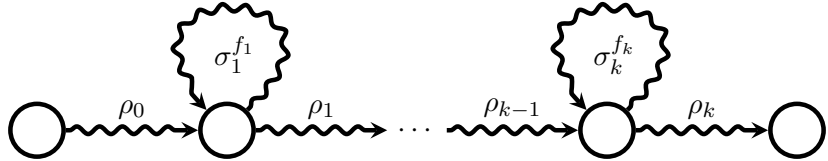
\begin{figure}[ht!]
    \begin{center}
    \begin{tikzpicture}
      \node[circle, draw, minimum size = 7mm, line width = 0.6mm] (q0) at (0, 0)    {};
      \node[circle, draw, minimum size = 7mm, line width = 0.6mm] (q1) at (2.5, 0)    {};
        \node (c1) at (2.6, 1) {$\sigma_1^{f_1}$};
      \node(q) at (5, 0) {$\cdots$};
      \node[circle, draw, minimum size = 7mm, line width = 0.6mm] (qk) at (7.5, 0)   {};
        \node (ck) at (7.6, 1) {$\sigma_k^{f_k}$};
      \node[circle, draw, minimum size = 7mm, line width = 0.6mm] (qf) at (10, 0)   {};
      \draw[-stealth, decoration={snake, amplitude = 0.5mm, segment length = 3mm, pre length = 0.5mm, post length=2mm}, decorate, line width = 0.6mm] 
      (q0) -- node[above]   {$\rho_0$}      (q1);
      \draw[-stealth, decoration={snake, amplitude = 0.5mm, segment length = 3mm, pre length = 0.51mm, post length=2mm}, decorate, line width = 0.6mm] 
      (q1) -- node[above]   {$\rho_1$}      (q);
      \draw[-stealth, decoration={snake, amplitude = 0.5mm, segment length = 3mm, pre length = 0.5mm, post length=2mm}, decorate, line width = 0.6mm] 
      (q)  -- node[above]   {$\rho_{k-1}$}  (qk);
      \draw[-stealth, decoration={snake, amplitude = 0.5mm, segment length = 3mm, pre length = 0.5mm, post length=2mm}, decorate, line width = 0.6mm] 
      (qk) -- node[above]   {$\rho_k$}      (qf);
      \path (q1) edge[-stealth, loop above, looseness = 25, distance = 70, out = 45, in = 135, decoration={snake, amplitude = 0.5mm, segment length = 3mm, pre length = 0.5mm, post length=2mm}, decorate, line width = 0.6mm] (q1);
      \path (qk) edge[-stealth, loop above, looseness = 25, distance = 70, out = 45, in = 135, decoration={snake, amplitude = 0.5mm, segment length = 3mm, pre length = 0.5mm, post length=2mm}, decorate, line width = 0.6mm] (qk);
    \end{tikzpicture}
    \end{center}
    \vspace{-2mm}
    \caption{A compressed linear form path.}
    \label{fig:compressed-linear-form}
\end{figure}

The following theorem is our main contribution.
\begin{theorem} 
    \label{thm:main-theorem}
    Let $V$ be a 2-VASS with one unary counter and fix two configurations
    $\config{p}{u}$ and $\config{q}{v}$. If $\run{\config{p}{u}}{\ast}{\config{q}{v}}$, then there exists a
    path in compressed linear form $\pi$ such that
    $\run{\config{p}{u}}{\pi}{\covfig{q}{v}}$ and $\vec{v}' \geq \vec{v}$.
    The size of the compressed linear form path is polynomial in $\abs{V} + \bit{\vec{u}} + \bit{\vec{v}}$.
\end{theorem}

\begin{corollary}
    \label{cor:coverability-in-np}
    Coverability in 2-VASS with one unary counter is in \class{NP}.
\end{corollary}
\begin{proof} 
By Theorem~\ref{thm:main-theorem} it suffices to consider paths in compressed linear form of polynomial size, that can be guessed in \class{NP}. 
It suffices to observe that a coverability instance on a given compressed linear form amounts to an instance of integer linear programming. 
Intuitively, this is because the nested cycles are fixed. 
Thus to check whether a run drops below zero it suffices to check before applying a cycle and after applying it for the last time (see \eg~\cite[Section V, Lemma 14]{BlondinFGHM15}).
\end{proof}

We highlight that it is rather unexpected that only one extra `level' of linear form paths is enough to obtain polynomial size witnesses of coverability in a 2-VASS with one unary counter, since the problem is \class{PSPACE}-complete for general 2-VASS. 
Roughly speaking, the example given in Figure~\ref{fig:example-main} observes the most complex behaviour possible and this instance of coverability is witnessed by a compressed linear form path.
More specifically, compressed linear form paths containing only one linear form cycle suffice as witnesses for coverability in 2-VASS with one unary counter. 
Therefore, all witnesses can be represented by a compressed linear form path $\rho \sigma^N \tau$ where $\rho$ and $\tau$ are linear form paths to and from the single linear form cycle $\sigma$ which is iterated $N$ times.

The rest of the paper is dedicated to proving Theorem~\ref{thm:main-theorem}. 
We heavily exploit both distinguishing features of the problem: the fact that one counter receives unary encoded updates (as opposed to both counters in binary) and the fact that we aim to assert coverability (as opposed to reachability). 
Our approach is as presented in the introduction.
In~\ref{sec:replacing} we observe that we can polynomially bound the total number of distinct short cycles.
We formalise this and show that there are only polynomially many `irreplaceable' short cycles.
In~\ref{sec:reshuffling} we provide a `reshuffling procedure'. 
If some short cycle $\gamma$ repeats exponentially many times we aim to modify the path $\pi$ by moving the cycles $\gamma$ close to each other.
Then either every short cycle $\gamma$ will appear only in polynomially many `bundles' $\gamma^e$, or we find a cycle $\sigma$ such that $\eff{\sigma} > \vec{0}$. In the latter case, by pumping $\sigma$ we are essentially left with one counter.
Finally, in Section~\ref{subsec:main-proof} we conclude the proof of Theorem~\ref{thm:main-theorem}.

\section{Replacing Short Cycles} \label{sec:replacing}

In this section, we show that there are only polynomially many short cycles that need occur in a run witnessing coverability. 
Fix a path $\pi = (q_{i-1}, b_i, u_i, q_i)_{i=1}^k$. 
Let $0 \leq i_b, i_u \leq k$ be the first indices such that $g_b = \sum_{i=1}^{i_b} b_i$ and $g_u = \sum_{i=1}^{i_u}u_i$ are at their lowest, respectively. 
Note that $g_b, g_u \leq 0$ since by convention if we consider $i_b, i_u = 0$ then the sum evaluates to $0$. 
We call and denote these two numbers the \emph{binary guard} $\binguard{\pi} = g_b$ and the \emph{unary guard} $\unaguard{\pi} = g_u$.
The following claim, that is proved in Appendix~\ref{app:replacing}, immediately follows from these definitions.

\begin{claim} \label{clm:guard-at-nadir}
Both $\binguard{\pi[i_b+1 .. k]} = 0$ and $\unaguard{\pi[i_u+1 .. k]} = 0$. 
\end{claim}

Much like the \emph{nadir} of a cycle in a one-counter net, defined in~\cite{AlmagorBHT20}, we define the \emph{binary-nadir state} as $q_{i_b}$, \ie the first state in which the binary counter first attains the lowest value when executing $\pi$. 
We call the \emph{binary-nadir decomposition} $\pi = \pi_1^b \pi_2^b$, for $\pi_1^b = \pi[1 .. i_b]$ and $\pi_2^b = \pi[i_b+1 .. k]$, as intimated in Claim~\ref{clm:guard-at-nadir}. 
Notice that this decomposition necessitates the binary guard of the path $\pi$ is equal to the binary effect of the prefix $\pi_1^b$, $\binguard{\pi} = \bineff{\pi_1^b} = \binguard{\pi_1^b}$.
Furthermore, the suffix of the binary-nadir decomposition has zero binary guard $\binguard{\pi_2^b} = 0$.
We primarily utilise binary-nadir states and binary-nadir decompositions, hence the omission of matching unary-nadir states and unary-nadir-decompositions. 

\begin{definition}[Replaceable cycles] \label{def:replaceable-cycle}
Let $\gamma$ be a $q$-cycle and let $p$ be the binary-nadir state of $\gamma$. 
We say that $\gamma$ is replaceable if there exists a $q$-cycle $\gamma'$ with the same binary-nadir state $p$, such that
\begin{enumerate}[(a)]
    \item $\bineff{\gamma'} \geq \bineff{\gamma}$ and $\unaeff{\gamma'} \geq \unaeff{\gamma}$,
    \item $\binguard{\gamma'} \geq \binguard{\gamma}$ and $\unaguard{\gamma'} \geq \unaguard{\gamma}$, and
    \item $\len{\gamma'} \leq \len{\gamma}$.
\end{enumerate}
Additionally, at least one inequality is strict and we write $\gamma \prec \gamma'$.
\end{definition}

We say a cycle is \emph{irreplaceable} if it is not replaceable.
We also say that an irreplaceable $q$-cycle $\gamma$ with the binary-nadir state $p$ is \emph{characterised} by the five values: $\bineff{\gamma}$, $\unaeff{\gamma}$, $\binguard{\gamma}$, $\unaguard{\gamma}$, and $\len{\gamma}$. 
The following lemma is proved in Appendix~\ref{app:replacing}.

\begin{lemma}[Replacing cycles] \label{lem:replacing}
Let $\pi = \pi_1 \gamma \pi_2$, where $\gamma$ is a $q$-cycle. 
Suppose $\run{\config{p}{u}}{\pi}{\config{q}{v}}$ then the following hold.
\begin{itemize}
    \item If $\gamma$ is replaceable, then there exists an irreplaceable $q$-cycle $\gamma \prec \gamma'$ such that $\run{\config{p}{u}}{\pi_1 \gamma' \pi_2}{\covfig{q}{v}}$.
    \item If $\gamma$ is irreplaceable, then for every irreplaceable $q$-cycle $\gamma'$ that has the same characterisation as $\gamma$, $\run{\config{p}{u}}{\pi_1 \gamma' \pi_2}{\covfig{q}{v}}$.
\end{itemize}
In both cases $\vec{v'} \geq \vec{v}$ and $\len{\pi} \geq \len{\pi_1 \gamma' \pi_2}$.
\end{lemma}

For convenience, we define the polynomial $R(\abs{Q}) \coloneqq \abs{Q}^4(\abs{Q}+1)(2\abs{Q}+1)^2$.

\begin{lemma} \label{lem:number-of-irreplaceable-cycles}
    There exists at most $R(\abs{Q})$ many irreplaceable short cycles with different characterisations.
\end{lemma}
\begin{proof} 
We fix two states $q$ and $p$ and consider only $q$-cycles $\gamma$ with the binary-nadir state $p$. 
Thus in the final argument one must multiply everything by $\abs{Q}^2$. 
Since we consider short cycles, the unary effect and the unary guard are small, \ie $-\abs{Q} \leq \unaeff{\gamma} \leq \abs{Q}$ and $-\abs{Q} \leq \unaguard{\gamma} \leq 0$.

Towards a contradiction, suppose there exists more than $\abs{Q}^2(\abs{Q}+1)(2\abs{Q}+1)^2$ many such irreplaceable $q$-cycles with different characterisations.
By the pigeonhole principle there must exist two cycles, denoted in binary-nadir decomposition $\gamma = \gamma_1\gamma_2$ and $\gamma' = \gamma'_1\gamma'_2$, that have the same values $\unaeff{\gamma_1} = \unaeff{\gamma'_1}$, $\unaeff{\gamma_2} = \unaeff{\gamma'_2}$, $\unaguard{\gamma} = \unaguard{\gamma'}$, $\len{\gamma_1} = \len{\gamma'_1}$, and $\len{\gamma_2} = \len{\gamma'_2}$.

We know that the irreplaceable $q$-cycles $\gamma$ and $\gamma'$ have different characterisations, so it must be the case that their binary effects differ $\bineff{\gamma} \neq \bineff{\gamma'}$. 
Otherwise, the cycle with the lesser binary guard is replaceable, because the unary effect, unary guard, and length do not differ. 
Without loss of generality, suppose $\bineff{\gamma} > \bineff{\gamma'}$, then $\binguard{\gamma} < \binguard{\gamma'}$. 
Otherwise, $\gamma'$ would be replaceable as $\gamma \prec \gamma'$. 

Now consider the $q$-cycle $\sigma = \gamma_1'\gamma_2$, also with the binary-nadir state $p$.
We will show that $\gamma \prec \sigma$ contradicting the fact that $\gamma$ is an irreplaceable $q$-cycle.
First, observe that $\sigma$ has greater binary effect than $\gamma$ as
\begin{equation*}
    \bineff{\sigma} = \bineff{\gamma_1'} + \bineff{\gamma_2} > \bineff{\gamma_1} + \bineff{\gamma_2} = \bineff{\gamma},
\end{equation*}
where the inequality holds because $\binguard{\gamma} < \binguard{\gamma'}$. Second, $\sigma$ and $\gamma$ have equal unary effect because $\unaeff{\gamma_1'} = \unaeff{\gamma_1}$.
Third, we show that $\sigma$ has a greater binary guard than $\gamma$. 
Since $\gamma_2$ is the suffix of the binary-nadir decomposition of $\gamma$, it must be true that $\binguard{\gamma_2} = 0$.
By Claim~\ref{clm:guard-at-nadir} $\binguard{\sigma} = \binguard{\gamma_1'}$. 
Combining these facts, $\binguard{\sigma} = \binguard{\gamma'} > \binguard{\gamma}$. 
Fourth, $\sigma$ has at least the unary guard of $\gamma$ because, in particular, the unary guard of the prefix of a path is at most the unary guard of the entire path.
\begin{align*}
    \unaguard{\sigma} 
    & \; = \min\set{\unaguard{\gamma'_1}, \unaeff{\gamma'_1} + \unaguard{\gamma_2}} \\
    & \; \geq \min\set{\unaguard{\gamma'}, \unaeff{\gamma'_1} + \unaguard{\gamma_2}} \\
    & \; = \min\set{\unaguard{\gamma}, \unaeff{\gamma_1} + \unaguard{\gamma_2}}
    = \unaguard{\gamma}
    .
\end{align*}
Fifth and finally, $\sigma$ and $\gamma$ have equal length because $\len{\gamma'_1} = \len{\gamma_1}$.
We have at least one strict inequality. 
Thus, we have reached the desired contradiction.
\end{proof}

\section{Reshuffling Linear Form Paths} \label{sec:reshuffling}

\subsection{Reshuffling Procedure}
There can be many linear forms for a path $\pi$. 
We will try to find an `optimal' one, so we introduce a cost function to
quantify linear forms. 
Recall that a linear form $\ell$ is a sequence of paths $\tau_0, \tau_1, \ldots, \tau_k$ and a sequence of cycles $\gamma_1, \ldots, \gamma_k$. 
If $\pi$ is in the linear form $\ell = (\tau_0, \gamma_1, \tau_1, \dots, \tau_{k-1}, \gamma_k, \tau_k)$ then we write $\pi_\ell = \tau_0 \gamma_1^{e_1} \tau_1 \cdots \tau_{k-1} \gamma_k^{e_k} \tau_{k}$, where $\pi = \pi_\ell$ (the index is here to stress the exact linear form). 
For this section, we will consider linear forms only containing short cycles $\gamma$, they will play a key role in the following arguments.

We define a cost function that assigns, to a linear form $\ell$, the following pair of naturals $C(\ell) \coloneqq \left( \sum_{i=0}^k \len{\tau_i}, k \right)$.
For convenience, we define the polynomial $P(\abs{Q}) \coloneqq 2(\abs{Q}^2 + 1)(\abs{Q}^2 + 2) \cdot R(\abs{Q})$, where $R$ is the polynomial defined for Lemma~\ref{lem:number-of-irreplaceable-cycles}. 
We say that a linear form $\ell$ is \emph{narrow} if $C(\ell) \leq (\abs{Q}(P(\abs{Q})+1), P(\abs{Q}))$, otherwise we say that $\ell$ is \emph{wide}. 
We say that the triple $(\pi', \sigma, \pi'')$ is a monotone cycle decomposition of a path $\pi$ if $\sigma$ is a monotone cycle, $\pi = \pi' \sigma \pi''$, and $\len{\sigma} < \len{\pi}$.

\begin{lemma}[Reshuffling] \label{lem:reshuffling}
Let $\pi$ be a path such that $\run{\config{p}{u}}{\pi}{\config{q}{v}}$. 
Then there exists a path $\rho$ such that $\run{\config{p}{u}}{\rho}{\config{q}{w}}$ where $\vec{w} \geq \vec{v}$,
$\len{\rho} \leq \len{\pi}$, and either 
\begin{enumerate}[(i)]
    \item there exists a narrow linear form for $\rho$, or
    \item there exists a monotone cycle decomposition of $\rho$.
\end{enumerate}
\end{lemma}
\begin{proof}
    We start with a series of preparations. 
    In the early part of this proof, we provide simple observations to ascertain some auspicious properties of our path.
    In the later part of this proof, we present the `reshuffling procedure' and conclude with one of the cases in the statement of this lemma.
    In this proof we will compare linear forms using the lexicographic order $\lex$, that is known to be a linear-order and a well-order.
    Formally,
    \begin{align*}
        C(\ell') \lex C(\ell) \iff  & (C(\ell'))_1 < (C(\ell))_1 \text{ or, } \\ 
                                    & (C(\ell'))_1 = (C(\ell))_1 \text{ and } (C(\ell'))_2 < (C(\ell))_2.
    \end{align*}
    
    We start with a path $\pi'$ such that $\run{\config{p}{u}}{\pi'}{\covfig{q}{v}}$ where $\vec{v}' \geq \vec{v}$, $\len{\pi'} \leq \len{\pi}$, and $\pi'$ has a linear form $\ell'$ that has the least cost among all linear forms for all like-paths.
    That means there does not exist another path $\pi''$ such that $\run{\config{p}{u}}{\pi''}{\configuration{q}{v}{''}}$ where $\vec{v}'' \geq \vec{v}$, $\len{\pi''} \leq \len{\pi}$, and $\pi''$ has a linear form $\ell''$ such that $C(\ell'') \lex C(\ell')$.

    For the first observation, suppose there exists $0 \leq i \leq k$ such that $\len{\tau_i} > \abs{Q}$. 
    Then the path $\tau_i$ can be written as $\tau_i = \tau' \gamma \tau''$, where $\gamma$ is a short cycle.
    We can define the linear form $\ell''$ by modifying $\ell'$ where $\tau_i$ is swapped for $\tau' \gamma \tau''$.
    Although this increments the number of cycles $k$, we decrease the total length of the paths as $\len{\tau'} + \len{\tau''} < \len{\tau_i}$ (recall that empty cycles are forbidden).
    Thus $C(\ell'') \lex C(\ell')$ contradicting the assumption that $\ell$ has minimum cost.
    Therefore, we assume that $\len{\tau_i} \leq \abs{Q}$ for all $0 \leq i \leq k$. 

    For the second observation, we define $U \coloneqq \set{0 \leq i \leq m: (\vec{v}_i)_2 < \abs{Q}}$ to be the set of indices of configurations in the run that have unary counter value less than $\abs{Q}$.
    Observe that if $\abs{U} > \abs{Q}^2 + 1$ then there are two indices $0 < i < j \leq m$ such that the two corresponding configurations in the run have matching states $q_{i} = q_{j}$ and equal unary counter values  $(\vec{v}_{i})_2 = (\vec{v}_{j})_2$.
    Then, regardless of sign of its binary effect, $\pi'[i..j]$ is a monotone cycle.
    Here, case (ii) immediately holds by decomposing $\pi'$ itself using the monotone cycle $\pi'[i .. j]$, given that $i > 0$ and $j \leq m$ implies $\len{\pi'[i .. j]} = j - i < m = \len{\pi'}$.
    Therefore, we assume $\abs{U} \leq \abs{Q}^2 + 1$. 
    We continue with the aim of satisfying the conditions of case (ii) by finding a monotone cycle decomposition.

    Let $d = \abs{\set{\gamma_1, \ldots, \gamma_k}}$ be the number of distinct cycles in the linear form $\ell'$.
    By Lemma~\ref{lem:replacing} and Lemma~\ref{lem:number-of-irreplaceable-cycles}, we can assume that $d \leq R(\abs{Q})$. 
    Otherwise, we can exchange replaceable $q$-cycles for irreplaceable $q$-cycles using the first point in Lemma~\ref{lem:replacing}.
    It is possible that for a particular characterisation, we can observe more than one irreplaceable $q$-cycle. 
    Then using the second point in Lemma~\ref{lem:replacing}, we can arbitrarily select one of these irreplaceable $q$-cycles with equal characterisations to exchange all others with.
    By applying these cycle replacements to $\pi'$, we obtain a different path $\rho$.
    Definition~\ref{def:replaceable-cycle} ensures that we do so without decreasing the effect (a), without allowing the counters to take a negative value (b), and without increasing the length of the path (c). 
    Therefore $\run{\config{p}{u}}{\rho}{\config{q}{w}}$ and $\vec{w} \geq \vec{v}' \geq \vec{v}$, and $\len{\rho} \leq \len{\pi'} \leq \len{\pi}$.
    We remark since cycles have been exchanged one-for-one, then $\rho$ takes a linear form $\ell$ with the same path segments as $\ell'$.
    Therefore, it is clear that neither the number of cycles $k$, nor the sum of the lengths of the paths between cycles, have changed.
    We also know that $\ell$ is a linear form for $\rho$ with minimum cost $C(\ell) = C(\ell')$, as per the initialisation in this proof.
    
    Suppose $\rho = \rho_\ell = \tau_0 \gamma_1^{e_1} \tau_1 \cdots \tau_{k-1} \gamma_k^{e_k} \tau_k$. 
    Let $(\configuration{q_j}{v}{_j})_{j=0}^m$ be the run obtained by following the path $\rho_\ell$ from the initial configuration $\configuration{q_0}{v}{_0} = \config{p}{u}$ to the final configuration $\configuration{q_m}{v}{_m} = \config{q}{w}$.
    We may assume that $\ell$ is wide.
    Otherwise, case (i) is immediately satisfied.
    We also know that $\len{\rho_\ell} \geq \max\set{(C(\ell))_1, (C(\ell))_2} > P(\abs{Q})$.
    We may also assume that each cycle $\gamma_1, \ldots, \gamma_k$ is non-monotone, \ie it is positive-negative or negative-positive.
    Otherwise, case (ii) immediately holds by decomposing $\rho$ itself using some monotone cycle $\gamma_i$, given that $\len{\gamma_i} \leq \abs{Q} < P(\abs{Q}) < \len{\rho_\ell}$.
    Notice this is valid since each $e_i > 0$ by the minimality of $C(\ell)$, otherwise you can write $\cdots\tau_{i-1} \gamma_i^0 \tau_i\cdots$ with one less cycle, decreasing $(C(\ell))_2$.
    
    From the first observation, we get $\sum_{i=0}^k \len{\tau_i} \leq (k+1) \abs{Q}$.
    Given that $\ell$ is wide, either $\abs{Q}(P(\abs{Q})+1) < (C(\ell'))_1 = \sum_{i=0}^k \len{\tau_i} \leq (k+1) \abs{Q}$ that implies $P(\abs{Q}) < k$, or $P(\abs{Q}) < (C(\ell'))_2 = k$.
    Regardless, $P(\abs{Q}) < k$ holds.
    Recall that $\abs{U} \leq \abs{Q}^2 + 1$ from the second observation.
    Since there are relatively `few' configurations indexed by $U$, there must exist a relatively `distant' pair of consecutive configurations indexed by $U$.
    More formally, there are $i$ and $j$ such that $0 \leq i < j \leq k$ and $j - i \ge 2(\abs{Q}^2+2)R(\abs{Q})$ and all configurations that occur in the run over the path segment $\tau_{i} \gamma_{i+1}^{e_{i+1}} \cdots \gamma_{j}^{e_j} \tau_{j}$ have unary counter value at least $\abs{Q}$.
    Notice that $j - i$ is the number of cycles in this path segment.
    Since $j - i \geq 2(\abs{Q}^2 + 2)R(\abs{Q})$ and by pigeonhole principle on the number of irreplaceable cycles, there is a common irreplaceable cycle $\gamma$ repeated at least $x = 2(\abs{Q}^2 + 2)$ many times. 
    We will focus on the first $x$ such occurrences of this cycle.
    Let $s_1, \ldots, s_x$ be the indices of this cycle $\gamma$, \ie $\gamma = \gamma_{s_1} = \ldots = \gamma_{s_x}$.
    To highlight these cycles, we decompose this path segment into
    \begin{equation*}
        \tau_{i} \gamma_{i+1}^{e_{i+1}} \cdots \gamma_{j}^{e_j} \tau_{j} = \Lambda_0 \gamma^{f_1} \Lambda_1 \cdots \Lambda_{x-1} \gamma^{f_x} \Lambda_x,
    \end{equation*}
    where $f_j \coloneqq e_{s_j}$ and $\Lambda_j$ are the concatenated paths (and cycles) in between iterations of $\gamma$, see Figure~\ref{fig:reshuffling0}.
    To reiterate, we know that all configurations that occur in the run over this path segment have at least $\abs{Q}$ unary counter value and $\gamma$ is a short cycle.
    \begin{figure}[ht!]
        \centering
        \includegraphics[width=\linewidth]{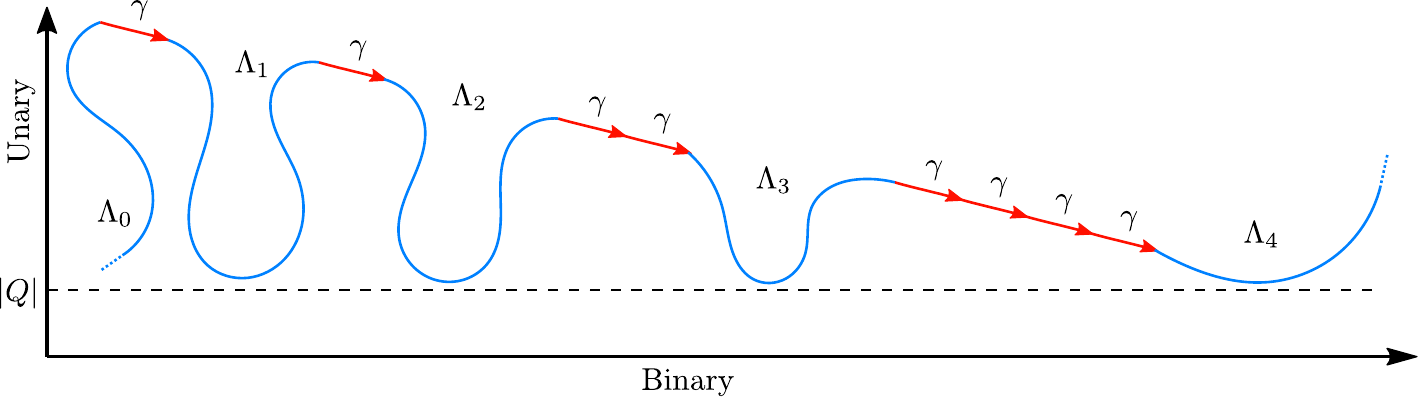}
        \caption{The decomposition of the path segment into $\Lambda_0 \gamma^{f_1} \Lambda_1 \cdots \Lambda_{x-1} \gamma^{f_x} \Lambda_x$, as above. 
        Notice that the unary counter is always at least $\abs{Q}$ as no configurations indexed by $U$ are present.}
        \label{fig:reshuffling0}
    \end{figure}

    \paragraph*{Reshuffling Procedure.}
    In the rest of the proof we will modify the path segment (above) of the path $\rho_\ell$ with a procedure that we call \emph{reshuffling}. 
    At the end of this procedure we will find a monotone cycle and satisfy case (ii) of this lemma.
    We either find this cycle directly, or we obtain a linear form $\ell''$ such that
    $C(\ell'') \lex C(\ell)$ contradicting the assumption that $\ell$ has
    minimal cost.
    
    Note that $x = 2(\abs{Q}^2 + 2)$ is even, and for every pair of consecutive cycles $\gamma_{2j-1}$ and $\gamma_{2j}$ (for $1 < 2j \leq x$), consider the subsegment $\gamma^{f_{2j-1}} \Lambda_{2j-1} \gamma^{f_{2j}}$. 
    There are two scenarios depending on the variant of the non-monotone cycle $\gamma$.
    In the scenario where $\gamma$ is positive-negative, we move an iteration of $\gamma$ from right to left obtaining $\gamma^{f_{2j-1}+1} \Lambda_{2j-1} \gamma^{f_{2j}-1}$. 
    In the scenario where $\gamma$ is negative-positive, we move an iteration of $\gamma$ in the opposite direction obtaining $\gamma^{f_{2j-1}-1} \Lambda_{2j-1} \gamma^{f_{2j}+1}$.

    We repeat this procedure until one of two conditions are met.
    The first is when there are no iterations of $\gamma$ on one side, so either $f_{2j-1}$ or $f_{2j}$ becomes $0$.
    The second is when there appears a configuration, in the run over the path subsegment after reshuffling, with unary counter value less than $\abs{Q}$.
    See Figure~\ref{fig:reshuffling1} for a pictorial presentation of reshuffling in the scenario where $\gamma$ is positive-negative.
    \begin{figure}[ht!]
        \vspace{3mm}
        \centering
        \includegraphics[width=\linewidth]{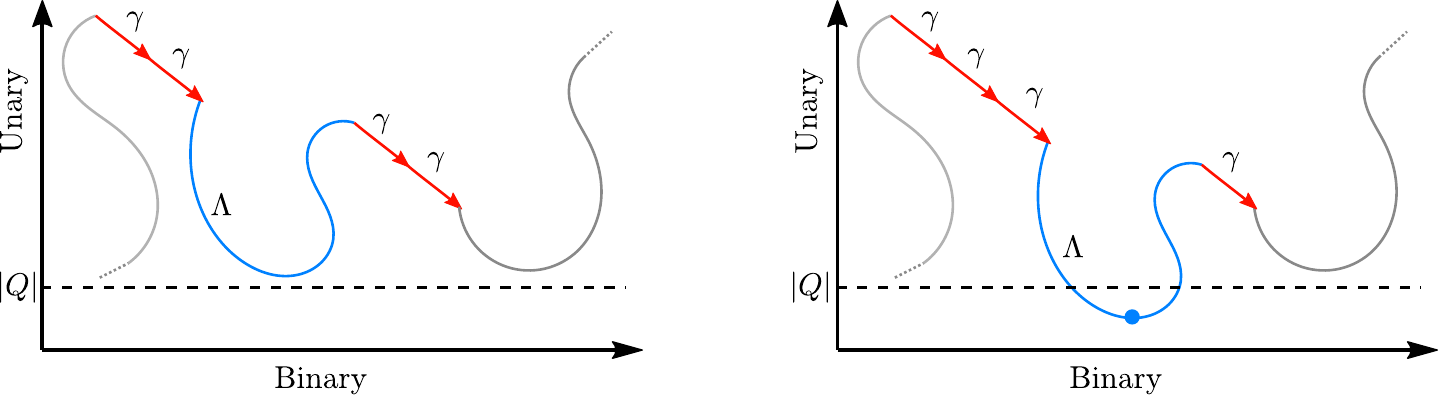}
        \vspace{1mm}
        \caption{Reshuffling around a path $\Lambda$ (blue) where $\gamma$ (red) is positive-negative.
        Before reshuffling, the path subsegment $\cdots \gamma \Lambda \gamma \cdots$ all configurations have unary counter value at least $\abs{Q}$ in the run (left).
        After reshuffling, the path subsegment $\cdots \gamma \gamma \Lambda \cdots$, there is a configuration with unary counter value less than $\abs{Q}$ in the run (right).
        }
        \label{fig:reshuffling1}
    \end{figure}

    We claim that after each reshuffling step, the corresponding run remains executable, so we must check that both counters remain non-negative.
    Notice that by only moving a cycle, the total effect of the path subsegment remains the same.
    Therefore, if the run was executable before reshuffling, we can safely assume that the prefix before the path subsegment and the suffix after the path subsegment are still executable.
    For that reason, consider the counter values of configurations occurring in the run over the reshuffled path subsegment.
    We focus on a single step of the reshuffling procedure that concerns the subsegment
    $\gamma^{f_{2j-1}} \Lambda_{2j-1} \gamma^{f_{2j}}$.

    Suppose $\gamma$ is a positive-negative cycle. 
    Then the reshuffling procedure moves $\gamma$ from right to left.  
    We claim that since $f_{2j-1} > 0$ and $\Lambda_0 \gamma^{f_1} \Lambda_1 \cdots \Lambda_{2j-1} \gamma^{f_{2j-1}}$ is executable, the subsegment $\Lambda_0 \gamma^{f_1} \Lambda_1 \cdots \Lambda_{2j-1} \gamma^{f_{2j-1}+1}$ is executable from the initial configuration.
    This is because one prerequisite of the reshuffling procedure is that all configurations occurring in the run over the path subsegment have at least $\abs{Q}$ unary counter value.
    Moreover, the cycle $\gamma$ has length at most $\abs{Q}$ so $\unaguard{\gamma} \geq -\abs{Q}$ means the unary counter value remains non-negative. 
    As for the binary counter value, since a single execution of $\gamma$ increases the binary counter and an iteration of $\gamma$ was already executed before reshuffling, $\Lambda_0 \gamma^{f_1} \Lambda_1 \cdots \Lambda_{2j-1} \gamma^{f_{2j-1}+1}$ is executable.  
    In the same way, from the initial configuration, $\Lambda_0 \gamma^{f_1} \Lambda_1 \cdots \Lambda_{2j-1} \gamma^{f_{2j-1}+1} \Lambda_{2j} \gamma_{2j}^{f_{2j}-1}$ is executable.
    This is because $\unaeff{\gamma} \geq -\abs{Q}$, and again, all configurations occurring in the run over the path subsegment have at least $\abs{Q}$ unary counter value, and also because of the monotonicity on the binary counter.

    The argument when $\gamma$ is a negative-positive cycle is analogous. 
    This concludes the correctness analysis of the reshuffling procedure.

    \paragraph*{Finishing Reshuffling.}
    We analyse what happens when reshuffling is finished.
    Suppose that there exists a pair $2j-1$ and $2j$ such that the reshuffling finishes under the first condition where all iterations of $\gamma$ have been moved to one side of $\Lambda_{2j-1}$.  
    In this case we obtain a new linear form $\ell''$ for $\rho$, where one collection of the cycle $\gamma$ has been removed (decrementing $k$).
    So $(C(\ell''))_2 = k - 1 < (C(\ell))_2$ and the two adjacent path segments can be combined without changing the summed length of paths so $(C(\ell''))_1 = (C(\ell))_1$.
    Therefore, $C(\ell'') \lex C(\ell)$ contradicting the assumption $\ell$ has the minimal cost.

    Otherwise, for every $1 \leq j \leq x/2$ the reshuffling of pair $2j-1$ and $2j$ finishes under condition the second condition.
    So there is a configuration with unary counter value less than $\abs{Q}$ in the run induced from the path $\rho$ for each pair $2j-1$ and $2j$ (see Figure~\ref{fig:reshuffling2}).
    Recall that $\frac{x}{2} = \abs{Q}^2 + 2$, that is the number of pairs.
    Akin to the first observation (in the beginning of this proof), we use the pigeonhole principle on the number of such configurations to obtain two configurations with matching states and equal unary counter values.  
    The path segment inducing the part of the run between these two configurations is a monotone cycle, regardless of the binary effect.
    Again, it must be true that the length of this cycle is less than the length of the whole path, so we obtain a monotone cycle decomposition of $\rho$.
    Thus case (ii) of the lemma holds.
    \begin{figure}[ht!] 
        \vspace{4mm}
        \centering
        \includegraphics[width=\linewidth]{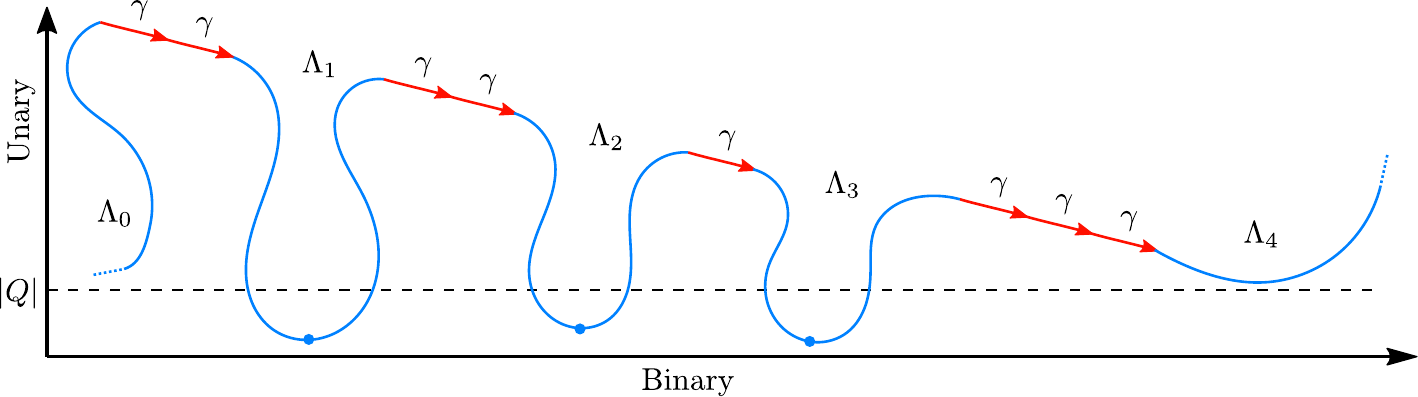}
        \caption{After reshuffling is finished under condition the second condition, we can find a zero unary effect cycle using the (sufficiently many) configurations with unary counter less $\abs{Q}$.}
        \label{fig:reshuffling2}
    \end{figure}
\end{proof}

\subsection{Applying Reshuffling} \label{subsec:reshuffling-application}
Lemma~\ref{lem:reshuffling} does not necessarily return a narrow linear form for a path $\pi$ witnessing coverability. 
Instead it may return a monotone cycle decomposition $(\rho, \sigma, \tau)$ of $\pi$.
Our next goal is to show that there exists polynomial size certificates for $\rho$ and $\sigma$ (Lemma~\ref{lem:decomposition}), and then to show that there exists a polynomial size certificate for $\tau$ (Lemma~\ref{lem:one-counter}).
Like linear forms, there can be many monotone cycle decompositions for a path. 
Following, we will use the cost function assigning monotone cycle decompositions to pairs of natural numbers $D((\rho, \sigma, \tau)) \coloneqq (\len{\rho\sigma}, \len{\sigma})$. 
Note that we can compare two decompositions using their cost, even if they are
for two different paths.

\begin{lemma} \label{lem:decomposition}
Suppose $\run{\config{p}{u}}{*}{\config{q}{v}}$ yet there is no narrow linear form $\ell$ for any path $\pi$ such that $\run{\config{p}{u}}{\pi}{\config{q}{w}}$ and $\vec{w} \geq \vec{v}$, then there exists a path $\pi'$ such that
\begin{enumerate}[(a)]
    \item $\run{\config{p}{u}}{\pi'}{\covfig{q}{w}}$ where $\vec{w}' \geq \vec{v}$,
    \item there is a monotone cycle decomposition $(\rho, \sigma, \tau)$ of $\pi'$ where $\eff{\sigma} > \vec{0}$, and
    \item there are narrow linear forms for both $\rho$ and $\sigma$.
\end{enumerate}
\end{lemma}
\begin{proof}
    We will again use the lexicographical order $\lex$ to compare the cost of monotone cycle decompositions.
    Let $\pi$ be a path of minimum length such that $\run{\config{p}{u}}{\pi}{\config{q}{w}}$ where $\vec{w} \geq \vec{v}$.
    Let $c = (\rho, \sigma, \tau)$ be the monotone cycle decomposition of $\pi$ that minimizes the cost $D(c)$ under the $\lex$ order. 
    Such a decomposition must exist, otherwise applying Lemma~\ref{lem:reshuffling} would return a narrow linear form $\ell'$ for $\rho$ such that $\run{\config{p}{u}}{\rho}{\covfig{q}{w}}$ and $\vec{w}' \geq \vec{w} \geq \vec{v}$, contradicting an assumption of this lemma.
    Observe that $\eff{\sigma} > \vec{0}$, otherwise one can remove $\sigma$ and consider the shorter path $\rho \tau$, contradicting the minimal length of $\pi$.
    Next, we argue that $\rho$ and $\sigma$ do not have monotone cycle
    decompositions, we then leverage Lemma~\ref{lem:reshuffling} to obtain the narrow linear forms required.

    \paragraph*{Path $\rho$ cannot be decomposed further.}
    Towards a contradiction, assume that there is a monotone cycle decomposition $c' = (\rho', \sigma', \tau')$ of $\rho$. 
    Observe that the following monotone cycle decomposition $c' = (\rho', \sigma', \tau' \sigma \tau)$ of $\pi$ has lower cost $D(c') \lex D(c)$ as $(D(c'))_1 = \len{\rho'} + \len{\sigma'} < \len{\rho} + \len{\sigma} = (D(c))_1$.
    This contradicts the assumption that $(\rho, \sigma, \tau)$ has minimum cost.

    Suppose $\run{\config{p}{u}}{\rho}{\config{p'}{x}}$.
    Since there is no monotone cycle decomposition, applying Lemma~\ref{lem:reshuffling} to $\rho$ returns a path $\rho'$ with a narrow linear form such that $\run{\config{p}{u}}{\rho'}{\covfig{p'}{x}}$ where $\vec{x}' \geq \vec{x}$ and $\len{\rho'} \leq \len{\rho}$.

    \paragraph*{Cycle $\sigma$ cannot be decomposed further.}
    Towards a contradiction, assume that there is a monotone cycle decomposition $(\rho', \sigma', \tau')$ of $\sigma$.
    Observe that the following monotone cycle decomposition $c' = (\rho \rho', \sigma', \tau' \tau)$ of $\pi$ has lower cost $D(c') \lex D(c)$ as $(D(c'))_1 = \len{\rho} + \len{\rho'} + \len{\sigma'} \leq \len{\rho} + \len{\sigma} = (D(c))_1$ and $(D(c'))_2 = \len{\sigma'} < \len{\sigma} = (D(c))_2$.
    This contradicts the assumption that $(\rho, \sigma, \tau)$ has minimum cost. 

    Suppose $\run{\config{p'}{x}}{\sigma}{\config{p'}{y}}$.
    Since there is no monotone cycle decomposition, applying Lemma~\ref{lem:reshuffling} to $\sigma$ returns a path $\sigma'$ with a narrow linear form such that $\run{\config{p'}{x}}{\sigma'}{\covfig{p'}{y}}$ where $\vec{y}' \geq \vec{y}$ and $\len{\sigma'} \leq \len{\sigma}$.
    In particular, it is also true that $\eff{\sigma'} \geq \eff{\sigma} > \vec{0}$.

    Replacing $\rho'$ for $\rho$ and $\sigma'$ for $\sigma$ in $\pi$ yields a path $\pi'$.
    Clearly if $\run{\config{p}{u}}{\pi}{\config{q}{w}}$ where $\vec{w} \geq \vec{v}$, then $\run{\config{p}{u}}{\pi'}{\config{q}{w'}}$ where $\vec{w}' \geq \vec{w} \geq \vec{v}$.
    Finally, $(\rho', \sigma', \tau)$ is monotone cycle decomposition of $\pi'$ such that $\eff{\sigma'} > \vec{0}$ and $\rho'$ and $\sigma'$ have narrow linear forms, as required.
\end{proof}

We now aim to obtain a narrow linear form for $\tau$. 
Note that Lemma~\ref{lem:decomposition} gives us a monotone cycle $\sigma$ with positive effect on at least one counter, \ie $\eff{\sigma} > \vec{0}$. 
By pumping $\sigma$ we can force one of the counters to take an arbitrarily large value (following, the vector $\vec{x}$ reflects this large value for Lemma~\ref{lem:one-counter}). 
Then, loosely speaking, the problem reduces to coverability in 1-VASS. 
However, proving the existence of a polynomial size compressed linear form path in Theorem~\ref{thm:main-theorem} requires more care. 
Note that Lemma~\ref{lem:one-counter} is stated for 2-VASS (not necessarily with one unary counter).
First we need to recall the following bound on counter values observed throughout runs.
Recall that $\abs{V}_{\max} \coloneqq \abs{Q} + \abs{T} \cdot \abs{T}_{\max}$ is the pseudo-polynomial size of the input.

\begin{theorem}[Corollary from Theorem 3.2 in \cite{BlondinEFGHLMT21}] \label{thm:blondin}
    Consider a 2-VASS (with both counters in binary) $V = (Q,T)$ and let $\run{\config{p}{u}}{*}{\config{q}{v}}$, then there exists a run $\config{p}{u} = \configvanilla{q_0}{\vec{v}_0}$, $\configvanilla{q_1}{\vec{v}_{1}}, \ldots, \configvanilla{q_m}{\vec{v}_m} = \config{q}{v}$ such that $\abs{\vec{v}_0}_{\max}, \abs{\vec{v}_1}_{\max}, \ldots,$ $\abs{\vec{v}_m}_{\max} \leq (\abs{V}_{\max} + \abs{\vec{u}}_{\max} + \abs{\vec{v}}_{\max})^{\Oh(1)}$.
\end{theorem}

In the following lemma, that is proved in Appendix~\ref{app:reshuffling}, given a 2-VASS $V$, the initial configuration $\config{p}{u}$, and target configuration $\config{q}{v}$, we write $B$ in place of $(\abs{V}_{\max} + \abs{\vec{u}}_{\max} + \abs{\vec{v}}_{\max})^{\Oh(1)}$ from Theorem~\ref{thm:blondin} and we fix $\vec{x} = (4B\abs{Q}^2\abs{V}_{\max}^2, 0)$. 
\begin{lemma}
    \label{lem:one-counter}
    Consider a 2-VASS (with both counters in binary) $V = (Q,T)$ and let $\run{\config{p}{u}}{*}{\config{q}{v}}$, then there exists a narrow linear form path $\pi'$ such that $\run{\config{p}{u + x}}{\pi'}{\covfig{q}{v}}$ for some $\vec{v}' \geq \vec{v}$.
\end{lemma}

\section{Proof of Theorem~\ref{thm:main-theorem}} \label{subsec:main-proof}

Before proving Theorem~\ref{thm:main-theorem}, we employ the fact that for a general 2-VASS, not necessarily with one unary counter, the exponents of cycles in linear forms can be pseudo-polynomially bounded.

\begin{lemma}[Corollary from Lemma 18 in~\cite{BlondinFGHM15}] \label{lem:poly-exponents}
    Let $\pi$ be path in a 2-VASS with a linear form $\pi = \tau_0 \gamma_1^{f_1}
    \tau_1 \ldots \gamma_k^{f_k} \tau_k$ such that $\run{\config{p}{u}}{\pi}{\config{q}{v}}$. 
    Then there exist a path $\pi' = \tau_0 \gamma_1^{e_1} \tau_1 \cdots \tau_{k-1} \gamma_k^{e_k} \tau_k$ such that $\run{\config{p}{u}}{\pi'}{\covfig{q}{v}}$ where $\vec{v}' \geq \vec{v}$ and $\bit{e_1}, \ldots, \bit{e_k}$ are all bounded by a polynomial in $\abs{V} + \bit{\vec{u}} + \bit{\vec{v}}$.
\end{lemma}

\begin{proof}[Proof of Theorem~\ref{thm:main-theorem}]    
    Let $\run{\config{p}{u}}{\pi}{\config{q}{v}}$ for some path $\pi$. 
    If there is a narrow linear form $\ell$ for $\pi$ then by Lemma~\ref{lem:poly-exponents} we obtain $\pi' = \tau_0 \gamma_1^{e_1} \tau_1 \cdots \tau_{k-1} \gamma_k^{e_k} \tau_k$ such that $\run{\config{p}{u}}{\pi'}{\covfig{q}{v}}$ where $\vec{v}' \geq \vec{v}$ and $\bit{e_1}, \ldots, \bit{e_k}$ are all bounded above by a polynomial in $\abs{V} + \bit{\vec{u}} + \bit{\vec{v}}$.
    Since $\ell$ is a narrow linear form, we know that $k \leq P(\abs{Q})$ so $\sum_{i = 1}^k \len{\gamma_i} \leq k\abs{Q} \leq \abs{Q}P(\abs{Q})$ and we also know that $\sum_{i=0}^{k} \len{\tau_i} \leq \abs{Q}(P(\abs{Q})+1)$.
    Together, this implies the linear form path $\pi'$ is of polynomial size.

    It remains to consider the case when there is no narrow linear form $\ell$ for $\pi$. By Lemma~\ref{lem:decomposition} (via Lemma~\ref{lem:reshuffling}) there exists a path $\pi'$ such that $\run{\config{p}{u}}{\pi'}{\covfig{q}{v}}$ and $\vec{v}' \geq \vec{v}$.
    Moreover, there is a monotone cycle decomposition $(\rho, \sigma, \tau)$ of $\pi'$ such that $\eff{\sigma} > \vec{0}$ and there are narrow linear forms for both $\rho$ and $\sigma$.

    Assume that $(\eff{\sigma})_1 > 0$. 
    This is without loss of generality because if $(\eff{\sigma})_1 = 0$ then one can flip the coordinates in $V$, $\vec{u}$ and $\vec{v}$ (for the remainder of the proof it will not matter that one counter is unary).
    Let $\config{p'}{m}$ be the configuration such that $\run{\config{p}{u}}{\;\rho\;}{\run{\config{p'}{m}}{\sigma\tau}{\covfig{q}{v}}}$.
    Observe that since $\eff{\sigma} > \vec{0}$ for every $i \in \naturals$ the path $\rho\sigma^i$ induces the run $\run{\config{p}{u}}{\rho\sigma^i}{\configuration{p'}{m}{+ i\cdot \eff{\sigma}}}$.  
    Consider $x = (\vec{x})_1 = 4B\abs{Q}^2\abs{V}_{\max}^2$ (for Lemma~\ref{lem:one-counter}), clearly $x$ is large enough so that $\run{\config{p}{u}}{\rho\sigma^x}{\covfig{p'}{m}}$ and $\vec{m}' \geq \vec{m} + \vec{x}$.
    By Lemma~\ref{lem:one-counter} there exists a narrow linear form for a path $\tau'$ such that $\run{\covfig{p'}{m}}{\tau'}{\configuration{q}{v}{''}}$ and $\vec{v}'' \geq \vec{v}'$.

    We conclude by considering the compressed linear form path $\rho \sigma^x \tau'$ such that $\run{\config{p}{u}}{\rho \sigma^x \tau'}{\configuration{q}{v}{''}}$ and $\vec{v}'' \geq \vec{v}' \geq \vec{v}$.
    Since $\rho$, $\sigma$, and $\tau'$ have narrow linear forms, we can also bound the exponents using Lemma~\ref{lem:poly-exponents} as in the beginning of this proof.
    Finally, $\bit{x}$ is polynomial in $\abs{V} + \bit{\vec{u}} + \bit{\vec{v}}$ much like the exponents of the cycles in the linear forms.
    Therefore, the size of the compressed linear form $\rho \sigma^x \tau'$ is polynomial in $\abs{V} + \bit{\vec{u}} + \bit{\vec{v}}$.
\end{proof}

\section{Conclusion and Future Work} \label{sec:conclusion}

In this paper we proved that coverability in 2-VASS with one unary counter is in \class{NP}, a drop in complexity from \class{PSPACE} for general 2-VASS.
We achieve this by using our new techniques.
Most notably, we polynomially bounded the number of short cycles that need to be used (Section~\ref{sec:replacing}).
Then, we attempt to find a polynomial linear form path by replacing short cycles and reshuffling the path (Section~\ref{sec:reshuffling}).

A natural extension is to consider whether coverability in 3-VASS with one binary counter and two unary counters is also in \class{NP}.
The technique for polynomially bounding the number of short cycles that need be used can easily be generalised to higher dimension VASS with one binary counter and multiple unary counters.
However, it is not clear how to modify and use our reshuffling technique.
Another open problem is whether reachability in 2-VASS with one unary counter is also in \class{NP}.
Note that completeness would immediately follow from the fact that reachability in binary encoded 1-VASS is \class{NP}-hard~\cite{HaaseKOW09}.

\bibliographystyle{plainurl}
\bibliography{references}

\begin{appendix}
    \section{Proofs of Section~\ref{sec:coverability}} \label{app:coverability}

\begin{claim} \label{clm:no-poly-linear-form-path-example}
    Given the instance of coverability consisting of $V$ from Figure~\ref{fig:example-main}, and the two configurations $q(0,1)$ and $q(N,1)$, there does not exist a path $\pi$ in linear form of polynomial size such that $\run{q(0,1)}{\pi}{q(\vec{v}')}$ where $\vec{v}' \geq (N,1)$.
\end{claim}
\begin{proof}
First, let us consider the scenario in a run from the configuration $q(0,1)$, this is the initial configuration.
The only option is to take the first transition $t_{qa}$ of $\lambda$, so let us consider a single iteration of this cycle.
So that the binary counter does not take a negative value, just before $t_{bq}$ the last transition in $\lambda$, its value must be at least $N^6-N$.
This only be satisfied when both $\alpha$ and $\beta$ are exhausted, so to speak.
In other words, $\alpha$ is iterated (exactly $N^2$ many times) until the value of the binary counter is less than $N^2$.
Then $\beta$ is iterated (also exactly $N^2$ many times) until the value on the unary counter is equal to $0$.
Therefore, an iteration of the cycle $\lambda = t_{qa}\alpha^{N^2}t_{ab}\beta^{N^2}t_{bq}$ is taken.
This has effect $(N, -1)$, so the configuration $q(N, 0)$ is observed.

Second, let us consider the scenario in a run from the configuration $q(N,0)$.
This occurs after following an iteration of $\lambda$, as detailed above.
The only option is to take the first transition $t_{qp}$ of $\rho$, so let us consider a single iteration of this cycle.
Like before, so that the binary counter does not take a negative value, just before $t_{dq}$ the last transition in $\rho$, its value must be at least $N^6$.
This is only satisfied when both $\gamma$ and $\delta$ are exhausted, much like for $\lambda$ above.
In other words, $\gamma$ is iterated (exactly $N^2$ many times) until the value of the binary counter is less than $N^2$.
Then $\delta$ is iterated (also exactly $N^2$ many times) until the value on the unary counter is equal to $0$.
Therefore, an iteration of the cycle $\rho = t_{qp}t_{pc}\gamma^{N^2}t_{cd}\delta^{N^2}t_{dq}$ is taken.
This has effect $(-N+1, 1)$, so the configuration $q(1,1)$ is observed.

We again are in the same scenario where the first transition of $\lambda$ can be iterated.
We repeat the above arguments, as a result alternating the cycles $\lambda$ and $\rho$ with the net effect $(1,0)$.
After $N$ iterations of the pair of cycles $\lambda\rho$, the configuration $q(N,1)$ is finally witnessed.

The number of iterations of each cycle $\alpha$ and $\beta$ within $\lambda$ is equal to $N^2$. 
The same is true for the number of iterations of each cycle $\gamma$ and $\delta$ within $\lambda$.
Additionally, in order to witness a configuration guaranteeing coverability of the target configuration $q(N,1)$, at least $N$ iterations of the pair $\lambda\rho$ are required.
Given that there are $4N$ such cycles in this run, it is not possible to construct a path in linear form of polynomial size.
\end{proof}

\section{Proofs of Section~\ref{sec:replacing}} \label{app:replacing}

\begin{proof}[Proof of Claim~\ref{clm:guard-at-nadir}]
    Suppose $\binguard{\pi[i_b+1 .. k]} < 0$.
    Then there exists $i_b+1 \leq j \leq k$ such that $\binguard{\pi[i_b+1 .. k]} = \sum_{i=i_b+1}^j b_i < 0$.
    Now consider the binary guard of the entire path $\pi$.
    \begin{equation*}
        \binguard{\pi} = \sum_{i=1}^{b_i} b_i > \sum_{i=1}^{b_i} b_i +
        \binguard{\pi[i_b+1 .. k]} = \sum_{i=1}^{j} b_i
    \end{equation*}

    The prefix path $\pi[1 .. j]$ has smaller effect than $\pi[1 .. i_b]$.
    This contradicts the definition of the binary guard of $\pi$, so
    $\binguard{\pi[i_b+1 .. k]} \geq 0$.
    We know that $\binguard{\pi[i_b+1 .. k]} \leq 0$, so we conclude
    $\binguard{\pi[i_b+1 .. k]} = 0$.
    The same argument can be used to show $\unaguard{\pi[i_u+1 .. k]} = 0$.
\end{proof}

\begin{proof}[Proof of Lemma~\ref{lem:replacing}]
    In either case, we refer to the Definition~\ref{def:replaceable-cycle}.
    By (a), we have $\bineff{\gamma'} \geq \bineff{\gamma}$ and $\unaeff{\gamma'} \geq \unaeff{\gamma}$ that imply $\vec{v}' \geq \vec{v}$.
    By (b), we have both $\binguard{\gamma'} \geq \binguard{\gamma}$ and $\unaguard{\gamma'} \geq \unaguard{\gamma}$ that imply the binary counter and the unary counter remain non-negative.
    By (c), $\len{\gamma'} \leq \len{\gamma}$, so we can conclude $\len{\pi} \geq \len{\pi_1 \gamma' \pi_2}$.
\end{proof}

\section{Proofs of Section~\ref{sec:reshuffling}} \label{app:reshuffling}

\begin{proof}[Proof of Lemma~\ref{lem:one-counter}]
    Since we work with a 2-VASS we write $\oneguard{\pi}$ and $\twoguard{\pi}$ (as opposed to $\binguard{\pi}$ and $\unaguard{\pi}$) to denote the guards on the first and second coordinate.
    Let $\run{\config{p}{u}}{\pi}{\config{q}{v}}$ such that $\pi$ is a bounded run as in Theorem~\ref{thm:blondin}.

    Somewhat similar to how it is defined in \cite{CzerwinskiLLP19}, we say a cycle $\sigma$ is \textit{semi-positive} if $(\eff{\sigma})_2 > 0$. 
    We say that $(\rho, \sigma, \tau)$ is a semi-positive cycle decomposition of $\pi$ if $\pi = \rho \sigma \tau$ and $\sigma$ is a semi-positive cycle. 

    Consider the case when there is no semi-positive cycle decomposition of $\pi$.
    Then from right to left we remove maximal length cycles from $\pi$, resulting in all cycles being removed from $\pi$, to obtain a simple path $\pi'$.
    Since we removed only cycles that are not semi-positive, observe that $\eff{\pi} \leq \eff{\pi'} + \vec{x}$ and $\twoguard{\pi} \leq \twoguard{\pi'}$.
    This holds because $\len{\pi'} < \abs{Q}$ (otherwise we could remove another cycle) and $\oneguard{\pi'} > - \abs{Q} \cdot \abs{V}_{\max}$.
    We can write the narrow linear form $\pi'$ as a single path satisfying the conditions of the lemma.

    We now focus on the case when $\pi$ contains a semi-positive cycle $\sigma$.
    Notice that we can assume that $\sigma$ is a simple cycle, otherwise $\sigma$ contains a shorter cycle $\sigma'$.
    If $\sigma'$ is a semi-positive cycle then either it is simple, or again we can  look for a shorter cycle in $\sigma'$. 
    If at some point we find a cycle that is not semi-positive, we delete it shortening  previous semi-positive cycles. 
    We continue this procedure exhaustively. 
    In the end, since every time we reduce the length of a cycle, we obtain a simple semi-positive cycle.

    Let $\rho\gamma$ be the shortest prefix of $\pi$ such that $\gamma$ is a simple semi-positive cycle, and suppose $\run{\config{p}{u}}{\rho}{\config{p'}{m}}$ (making $\gamma$ a $p'$-cycle).
    Such a $\rho$ and $\gamma$ must exist by the observation in the previous paragraph. 
    Now, let $\rho'$ be the simple path constructed by consecutively removing simple cycles from $\rho$.
    By minimality of $\rho$, only cycles that are not semi-positive are removed. 
    Recall that by Theorem~\ref{thm:blondin}, $\abs{\vec{m}}_{\max} \leq B$.
    Thus $\run{\configvanilla{p}{\vec{u} + (B(1 +
    \abs{V}_{\max}),0)}}{\rho'}{\covfig{p'}{m}}$ for some $\vec{m}' \ge \vec{m}$.

    Let $\tau$ be a simple path from $p'$ to $q$. 
    Such a path exists as $\pi$ is a path from $p$, through $q'$, to $q$. 
    We define the path $\pi' \coloneqq \rho' \gamma^e \tau$, where $e \coloneqq B\cdot  \abs{Q} \cdot \abs{V}_{\max}$.
    This provides a narrow linear form for $\pi'$ as $\rho'$ is a simple path, $\gamma$ is a simple cycle, and $\tau$ is a simple path.
    Since $\gamma$ is simple cycle, it is short, thus both $\oneguard{\gamma} > -\abs{Q}\cdot \abs{V}_{\max}$ and $(\eff{\gamma^e})_1 > -e \abs{Q} \cdot \abs{V}_{\max} $ are bounded below by $-B \abs{Q}^2 \abs{V}_{\max}^2$.
    Thus $\run{\config{p}{u + x}}{\rho'\gamma^e}{\configuration{p'}{m}{''}}$, where $\vec{m}'' \geq \vec{v} + \frac{\vec{x}}{2}$. 
    Finally, since $\tau$ is simple $\oneguard{\tau}, (\eff{\tau})_1 \geq -\abs{Q} \cdot \abs{V}_{\max}$. 
    Thus $\run{\config{p}{u + x}}{\pi'}{\covfig{q}{v}}$ for $\vec{v}' \geq
    \vec{v}$.
\end{proof}
\end{appendix}

\end{document}